\documentclass[11pt]{article}

\def\colorful{0}

\usepackage[utf8]{inputenc} %
\usepackage[T1]{fontenc}    %
\usepackage{hyperref}       %
\usepackage{url}            %
\usepackage{booktabs}       %
\usepackage{amsfonts}       %
\usepackage{nicefrac}       %
\usepackage{microtype}      %
\usepackage{xcolor}         %

\usepackage[T1]{fontenc}
\usepackage{microtype}
\usepackage[colorinlistoftodos]{todonotes}
\usepackage{dsfont}
\usepackage{environ}
\usepackage{amsthm,amsfonts,amsmath,amssymb,epsfig,color,float,graphicx,verbatim, enumitem}
\usepackage{multirow}
\usepackage{algorithm}
\usepackage{algorithmicx}
\usepackage[noend]{algpseudocode}
\usepackage{caption}
\usepackage{subcaption}
\usepackage{threeparttable}
\usepackage{hyperref}
\usepackage{enumitem}

\usepackage[utf8]{inputenc}

\usepackage{xspace}                                     %
\usepackage{thm-restate}
\usepackage{dsfont} %

\usepackage{algorithmicx,algpseudocode,algorithm}

\usepackage{csquotes}

\usepackage{relsize}
\usepackage[style=alphabetic,maxbibnames=99,minalphanames=99,maxalphanames=99,maxcitenames=99,backend=biber]{biblatex}

\usepackage{multirow}
\usepackage{chngpage} %

\usepackage{mfirstuc}

\usepackage{tikz}
\usetikzlibrary{arrows}
\usetikzlibrary{calc,decorations.pathmorphing,patterns}

\usepackage{aliascnt}
\usepackage[numbered]{bookmark}

\usepackage{bbm}
\usepackage{verbatim}
\usepackage{mleftright} %

\usepackage{environ}

\usepackage{mathtools}
\usepackage{amsthm,amsfonts,amsmath,epsfig,color,float,graphicx,verbatim, enumitem}
\usepackage{multirow}
\usepackage{caption}
\usepackage{subcaption}
\usepackage{threeparttable}
\usepackage{dsfont} %

\usepackage{times}
\usepackage{amsthm,amsfonts,amsmath,amssymb,epsfig,xcolor,graphicx,verbatim}
\usepackage{multirow}

\usepackage{enumitem}

\usepackage{framed}
\usepackage{nicefrac}

\usepackage{dsfont}
\newtheorem{theorem}{Theorem}[]

\newtheorem{lemma}{Lemma}[]
\newtheorem{informal theorem}[theorem]{Theorem (informal statement)}

\newtheorem{proposition}[theorem]{Proposition}

\newtheorem{claim}[theorem]{Claim}
\newtheorem{fact}[theorem]{Fact}

\newtheorem{definition}{Definition}

\newcommand{\eqdef}{\coloneqq}
\newcommand{\R}{\ensuremath{\mathbb R}}
\newcommand{\Z}{\ensuremath{\mathbb Z}}
\newcommand{\N}{\ensuremath{\mathbb N}}
\newcommand{\F}{\ensuremath{\mathcal F}}

\newcommand{\C}{\ensuremath{\mathcal C}}

\newcommand{\poly}{\operatorname{poly}}

\newcommand{\E}{\mathbf{E}}

\newcommand{\junk}[1]{}

\newcommand{\eps}{\varepsilon}

\newcommand{\dtv}{\operatorname{TV}}
\newcommand{\totalvardist}[2]{\totalvardistrestr{}{#1}{#2}}
\newcommand{\totalvardistrestr}[3]{\dtv^{#1}\mleft(#2, #3\mright)}

\DeclareMathOperator{\Median}{median}
\newcommand{\pmap}{ \hat{\varphi} }
\newcommand{\DY}{D_{\text{YES}}}
\newcommand{\DN}{D_{\text{NO}}}

\newcommand\snorm[2]{\left\| #2 \right\|_{#1}}
\newcommand{\norm}[1]{\lVert#1{\rVert}}
\newcommand{\normone}[1]{{\norm{#1}}_1}

\newcommand{\abs}[1]{\left | #1 \right |}

\newcommand{\lp}{\left}
\newcommand{\rp}{\right}
\newcommand\dchi[3]{ d_{\chi^2}^{#1}\lp( #2 \big\| #3\rp) }

\def \h {\mathbf h}

\newcommand{\B}{\mathcal B}
\newcommand{\D}{\mathcal D}
\newcommand{\I}{\mathcal I}
\newcommand{\wt}{\widetilde}

\newlength\myindent
\newcommand{\A}{\mathcal A}
\newcommand{\G}{\mathcal G}
\newcommand{\Q}{\mathcal Q}
\newcommand{\J}{\mathcal J}
\newcommand{\TV}{\text{TV}}
\newcommand{\bs}{\backslash}
\newcommand{\p}{\mathbf p}
\newcommand{\q}{\mathbf q}
\renewcommand{\S}{\mathcal S}
\DeclareMathOperator{\Poi}{Poi}
\def \Bad {\mathcal B}

\newcommand{\clg}[1]{\left\lceil #1 \right\rceil}

\newcommand{\proba}{\Pr}
\newcommand{\probaOf}[1]{\proba\!\left[\, #1\, \right]}

\newcommand{\expect}[1]{\mathbf{E}\!\left[#1\right]}

\newcommand{\poisson}[1]{\ensuremath{\operatorname{Poi}\!\left( #1 \right) }}

\providecommand{\poly}{\operatorname*{poly}}

\newcommand{\dst}{\eps}

\newcommand{\ns}{m}

\newcommand{\w}{\mathbf{w}}

\ifnum\colorful=1

	\newcommand{\inote}[1]{\footnote{{\bf [[Ilias: {#1}\bf ]] }}}
	\newcommand{\dnote}[1]{\footnote{{\bf [[Daniel: {#1}\bf ]] }}}
	\newcommand{\cnote}[1]{\footnote{{\bf [[Clement: {#1}\bf ]] }}}
	\newcommand{\snote}[1]{\footnote{{\bf [[Sihan: {#1}\bf ]] }}}
\else

	\newcommand{\inote}[1]{}
	\newcommand{\dnote}[1]{}
	\newcommand{\cnote}[1]{}
	\newcommand{\snote}[1]{}
\fi

\newcommand{\algname}[1]{\textup{\sc{}#1}}

\title{Near-Optimal Bounds for Testing Histogram Distributions}

\author{
Cl\'ement L. Canonne\\
University of Sydney\\
{\tt clement.canonne@sydney.edu.au} \\
\and
Ilias Diakonikolas\thanks{Supported by NSF Medium Award CCF-2107079,
NSF Award CCF-1652862 (CAREER), a Sloan Research Fellowship, and
a DARPA Learning with Less Labels (LwLL) grant. Some of this work was performed while the author
was visiting the Simons Institute for the Theory of Computing.}\\
University of Wisconsin-Madison\\
{\tt ilias@cs.wisc.edu}\\
\and
Daniel M. Kane\thanks{Supported by NSF Medium Award CCF-2107547,
NSF Award CCF-1553288 (CAREER), and a Sloan Research Fellowship.}\\
University of California, San Diego\\
{\tt dakane@cs.ucsd.edu}\\
\and 
Sihan Liu\thanks{Some of this work was performed while the author was an undergraduate student 
at UW Madison, supported by Ilias Diakonikolas' Sloan Fellowship.}\\
University of California, San Diego\\
{\tt sil046@ucsd.edu} \\
}

\begin{document}

\maketitle

\begin{abstract}
We investigate the problem of testing whether a discrete probability distribution over an ordered domain 
is a histogram on a specified number of bins. 
One of the most common tools for the succinct approximation of data, 
{\em $k$-histograms} over $[n]$, are probability distributions 
that are piecewise constant over a set of $k$ intervals. 
The histogram testing problem is the following:
Given samples from an unknown distribution $\p$ on $[n]$, 
we want to distinguish between the cases that $\p$ is a $k$-histogram 
versus $\eps$-far from any $k$-histogram, in total variation distance. 
Our main result is a sample near-optimal and computationally efficient algorithm 
for this testing problem, and a nearly-matching (within logarithmic factors) sample complexity lower bound. 
Specifically, we show that the histogram testing problem has sample complexity 
$\widetilde \Theta (\sqrt{nk} / \eps + k / \eps^2 + \sqrt{n} / \eps^2)$.
\end{abstract}

\setcounter{page}{0}
\thispagestyle{empty}
\newpage

\section{Introduction} \label{sec:intro}

\subsection{Background and Motivation} \label{ssec:background}

A classical approach for the efficient exploration of massive datasets involves the construction 
of succinct data representations, see, e.g., the survey~\cite{CGHJ12}. 
One of the most useful and commonly used compact representations are {\em histograms}. 
For a dataset $S$, whose elements are from the universe $[n] \eqdef \{1, \ldots, n\}$, a $k$-histogram 
is a function that is piecewise constant over $k$ interval pieces. 
Histograms constitute the oldest and most popular synopsis structure 
in databases and have been extensively studied in the database community 
since their introduction in the 1980s~\cite{Kooi:1980}, 
see, e.g.,~\cite{GMP97, JPK+98, CMN98, TGIK02, GGI+02, GKS06, ILR12, ADHLS15, Canonne16}, 
for a partial list of references.  In both the statistics and computer science literatures, 
several methods have been proposed to estimate histogram distributions 
in a range of natural settings~\cite{Scott79, FreedmanD1981,Devroye2004, LN96, 
Klem09, ChanDSS13, ChanDSS14, CDSS14b, ADHLS15, ADLS17, DLS18}.

In this work, we study the algorithmic task of deciding whether a (potentially very large) dataset $S$ 
over the domain $[n]$ is a $k$-histogram (i.e., it has a succinct histogram representation with $k$ interval pieces) 
or is ``far'' from {\em any} $k$-histogram representation (in a well-defined technical sense). 
Our focus is on {\em sublinear time } algorithms~\cite{Rubinfeld:06survey}. 
Instead of reading the entire dataset $S$, which could be highly impractical, 
one can instead use randomness to sample a small subset of the dataset. 
Note that sampling a (uniformly) random element from $S$ is equivalent 
to drawing a sample from the underlying probability distribution $\p$ 
of relative empirical frequencies. This observation brings our algorithmic problem 
of ``histogram testing'' in the framework of distribution property testing 
(statistical hypothesis testing)~\cite{BFR+:00, Batu13}; 
see, e.g.,~\cite{Canonne15-survey} for a survey.

For an integer $1\leq k\leq n$, denote by $\mathcal{H}_k^n$ 
the set of $k$-histogram distributions over $\{1,2,\dots,n\}$, 
i.e., the set of all distributions $\p$ such that there exists a partition of $[n]$ 
into $k$ consecutive intervals (not necessarily of the same size) 
with $\p$ being uniform on each interval. 
Formally, we study the following task: 
Given access to i.i.d.\ samples
from an unknown distribution $\p$ on $[n]$ and a desired error tolerance $0<\eps < 1$, 
we want to correctly distinguish (with high probability) between the cases that $\p$ 
is a $k$-histogram versus $\eps$-far from any $k$-histogram, 
in total variation distance (or, equivalently, $\ell_1$-norm). 
It should be noted that the histogram testing problem studied here is not new. 
Prior work within the algorithms and database theory communities 
has investigated the complexity of the problem in the past ten years 
(see, e.g.,~\cite{ILR12, ADHLS15, Canonne16} and 
Section~\ref{ssec:prior} for a detailed summary of prior work). 
However, known algorithms for this task are highly suboptimal;
in particular, there is a polynomial gap between the best known upper and lower bounds 
on the sample complexity of the problem. At a high level, the difficulty of our histogram 
testing problem in the sublinear regime lies in the fact that the location and ``length'' 
of the $k$ intervals defining the histogram representation (if one exists) 
is a priori unknown to the algorithm.

We believe that the histogram testing problem is natural and interesting in its own right. 
Moreover, a sample-efficient algorithm for this testing task can be used 
as a key primitive in the context of {\em model selection}, 
where the goal is to identify the ``most succinct'' data representation.  
Indeed, various algorithms are known for learning $k$-histograms from samples 
whose sample complexities (and running times) scale proportionally 
to the succinctness parameter $k$ (and are completely independent 
of the domain size $n$)~\cite{CDSS14b, ADHLS15, ADLS17}. 
Combined with an efficient tester for the property of being a $k$-histogram 
(used to identify the smallest possible value of $k$ 
such that $\p$ is a $k$-histogram, e.g., via binary search), 
one can obtain a sketch of the underlying dataset. 
See Appendix~\ref{sec:msel} for a detailed description.

\subsection{Our Results} \label{ssec:results}
Our main contribution is a near-characterization of the sample complexity of the histogram testing problem. 
Specifically, we provide (1)~a sample near-optimal and computationally 
efficient testing algorithm for the problem, 
and (2)~a nearly-matching sample complexity lower bound (within logarithmic factors).
Specifically, we establish the following theorem:

\begin{theorem}[Main Result] \label{thm:main}
There exists a testing algorithm for the class $\mathcal{H}_k^n$ of $k$-histograms on $[n]$ 
with sample complexity $m = \wt O( \sqrt{nk}  / \eps + k/\eps^2 + \sqrt{n}/\eps^2)$ and running time $\poly(m)$.
Moreover, for any $k \in [n]$ {and $0<\eps<1$}, any testing algorithm for the class $\mathcal{H}_k^n$ 
requires at least $\wt \Omega( \sqrt{nk}  / \eps + k/\eps^2 + \sqrt{n} / \eps^2  )$ samples.
\end{theorem}

\noindent (The $\tilde{O}(\cdot)$ and $\tilde{\Omega}(\cdot)$ notation hides polylogarithmic factors in the argument.) 
Theorem~\ref{thm:main} thus characterizes the complexity of the histogram testing problem 
within polylogarithmic factors. 

Note that there are three terms in the sample complexity; namely, $\sqrt{nk} / \eps$, $k / \eps^2$, 
and $\sqrt{n} / \eps^2$. 
The sample complexity of the problem is dominated by one of these three different terms, 
depending on the relative sizes of $n, k$ and $1/\eps$. 
An illustration is given in Figure~\ref{fig:sc}.

\begin{figure}[htp]
\includegraphics[width=10cm]{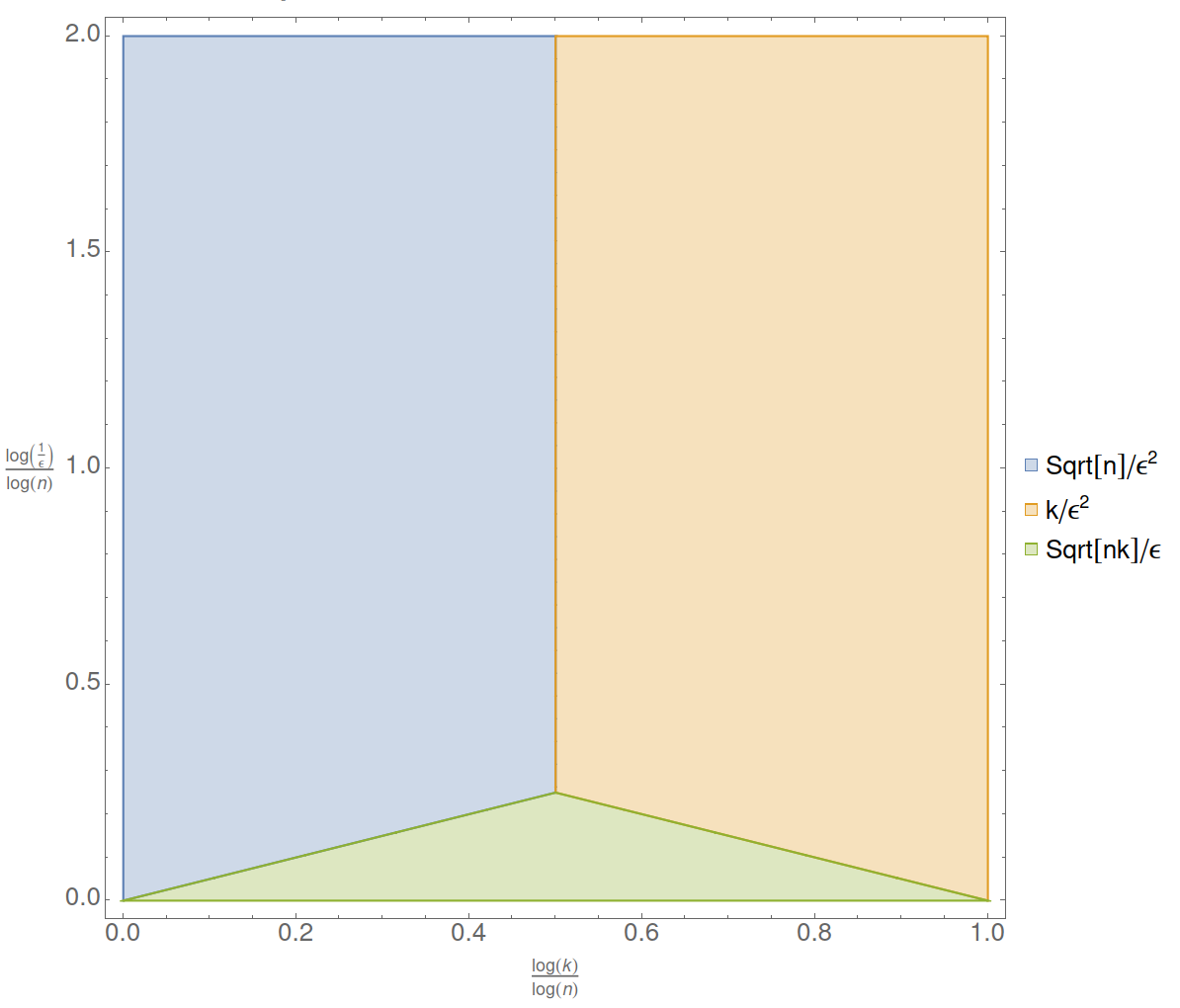}
\centering
\caption{\small{}The $x$-axis and $y$-axis represent the quantities 
$\log(k) / \log(n) $ and $\log(1/\eps)/ \log(n)$ respectively. 
Each point in the graph corresponds to a setting of the parameters $n, k, \eps$, 
and is colored based on the corresponding dominating term.}
\label{fig:sc}
\end{figure}

The best previous histogram testing algorithm had sample complexity 
$\wt O(\sqrt{kn}/\eps^3)$~\cite{CanonneDGR18}, 
while the best known lower bound was~$\wt \Omega (\sqrt{n} / \eps^2 + k /\eps)$~\cite{Canonne16}.
\footnote{As discussed in Section~\ref{ssec:prior}, while an upper bound of $\wt{O}(\sqrt{n} / \eps^2 + k /\eps^3)$ 
is claimed in~\cite{Canonne16}, the analysis of their algorithm\nopagebreak{} is flawed; 
and, indeed, our work shows that the upper bound stated in~\cite{Canonne16} \emph{cannot} hold, 
as it would contradict our\nopagebreak{} lower bound.\nopagebreak}
{We note that the previously best known upper and lower bounds exhibit a polynomial gap, 
even for constant values of~$\eps$ or $k$. 
For example, in the ``large-$k$'' regime, where $k = n^c$ for some constant $0< c<1$, 
there was a gap between $\wt O(n^{1/2+c/2})$ 
and $\tilde{\Omega} (n^{1/2})$ in the sample complexity. 
In this regime, Theorem~\ref{thm:main} gives the near-optimal bound of $\wt \Theta(n^{1/2+c/2})$.} 
{Similarly, in the ``high-accuracy'' regime, 
where $\eps=1/n^c$ for some constant $c>0$ (and, say, constant $k$), 
previous bounds imply that the sample complexity lies between 
$\wt O(n^{1/2+3c})$ and $\wt\Omega(n^{1/2+2c})$, 
while our result again gives the nearly-tight bound of $\wt \Theta(n^{1/2+2c})$.} 
These are only two specific examples: more generally, the previously known bounds 
are suboptimal by polynomial factors in $1/\eps$ when \smash{$\eps \geq \sqrt{k/n}$}; 
and {by polynomial factors} in all parameters $k,n,1/\eps$ 
when \smash{$\eps \leq \sqrt{k/n}$}. Theorem~\ref{thm:main} settles the sample complexity of the problem, 
up to logarithmic factors, for \emph{every} parameter setting.

At a technical level, our sample complexity lower bound construction conceptually 
differs from previous work in distribution testing, drawing instead from sophisticated 
techniques from the distribution \emph{estimation} literature. 
{Our upper bound leverages the ``Testing-via-Learning'' framework proposed in \cite{AcharyaDK15}. 
The main technical innovation enabling our algorithm is a sample and computationally efficient 
\emph{adaptive} algorithm which can simultaneously (1) learn an unknown histogram distribution 
\emph{with unknown interval structure}, and (2) identify a domain where the learned result is accurate.} 
We elaborate on these aspects next.

\subsection{Overview of Techniques} \label{ssec:techniques}

\paragraph{Sample Complexity Lower Bound.} 
We begin by discussing our lower bound techniques.
We follow the typical high-level approach in 
proving sample complexity lower bounds in this setting. 
Namely, we define two ensembles of distributions $\DY$ and $\DN$ 
such that, with high probability, the following conditions are satisfied: 
(1) a random distribution from $\DY$ is a $k$-histogram, 
(2) a random distribution from $\DN$ is $\eps$-far from any $k$-histogram, 
and (3) given samples of {appropriate size}, it is information-theoretically impossible 
to distinguish a random distribution drawn from $\DY$ from a random distribution drawn from $\DN$. 

{We start by describing our hard instances for the case that the accuracy parameter $\eps$ 
is a small universal constant.}
On the one hand, we define $\DY$ so that all $\p_i$'s are the same, 
except for a ``small'' number of domain elements, i.e., $c \cdot k$ for a small constant $c \in (0, 1)$. 
On the other hand, for a distribution $\p$ drawn from $\DN$, 
$\p_i$ will be randomly $0$ or roughly $2/n$, 
except for at most a constant fraction of the {elements}. 
It is not hard to see that, with high probability, a distribution drawn from $\DY$ (resp.\ $\DN$) 
will be a $k$-histogram (resp.\ far from being a $k$-histogram).

To ensure that the underlying distributions are indistinguishable using a small sample size, 
we want to {guarantee} that for all small values of $t$ 
the number of {elements} with exactly $t$ samples 
will be roughly the same for $\DY$ and $\DN$; 
this property rules out any test statistic relying on counting 
the number of $t$-way collisions among the samples. 
Following \cite{Valiant11,valiant2013estimating,Jiao15,Wu2016},
this is essentially equivalent to showing 
that distributions drawn from $\DY$ and $\DN$ \emph{match {their low-degree} moments}. 
In particular, for a random pair of distributions $\p$, $\p'$ drawn from $\DY$ and $\DN$ respectively,
we want that $\sum_i \p_i^t$ and $\sum_i {\p'}^t_i$ are roughly the same {for all small values $t$}.
We note that the non-exceptional {elements} of a distribution {$\p'$ drawn from} $\DN$ 
--- which {have probability mass either $0$ or roughly $2/n$} ---
will have second moment larger than the non-exceptional {elements} of a distribution {$\p$ drawn from} $\DY$ 
--- which {have probability mass roughly $1/n$} --- by approximately $1/n$. 
To counteract this {discrepancy}, 
the (fewer than $k$) exceptional {elements} in $\DY$ 
must have average mass at least $ 1/\sqrt{kn}$.
Fortunately, using techniques from~\cite{valiant2013estimating,wu2019chebyshev}, 
we are able to {construct} distributions that match $t = \Theta(\log n)$ moments
in which no individual bin has mass more than $\widetilde O(1/\sqrt{kn})$. 
Combining this construction with basic information-theoretic arguments 
gives us an $\widetilde \Omega(\sqrt{kn})$ sample complexity lower bound. 
We note that this lower bound is tight in the sense that 
with more than $\widetilde \Omega( \sqrt{kn})$ 
samples one can reliably identify the exceptional {elements}, 
as they will each have relatively large numbers of samples with high probability; 
this {allows us to distinguish $\DY$ from $\DN$ simply based 
on the subdistributions over these {elements}}.

{Given the aforementioned construction (for constant $\eps$)}, 
it is easy to obtain a sample lower bound of $\widetilde \Omega( \sqrt{kn}/\eps)$ 
by mixing our hard instances with {the} uniform distribution 
(with mixing weights $\eps$ and $1-\eps$ respectively). In fact, 
even if {the testing algorithm knows in advance} 
which samples come from the uniform part and which samples 
come from the original hard instance, distinguishing would still require $\widetilde \Omega(\sqrt{kn})$ 
samples from the original hard instance, 
and therefore $\widetilde \Omega(\sqrt{kn}/\eps)$ samples overall.
This sample size lower bound turns out to be tight for $\eps$ relatively large, 
as one can still reliably identify the exceptional bins with only $\widetilde \Omega(\sqrt{kn}/\eps)$ samples.
However, when $\eps$ becomes sufficiently small, identifying the exceptional bins becomes more challenging.
Indeed, if we take $m$ samples, we expect that an exceptional bin has roughly 
$m \eps/\sqrt{kn}$ more samples than a non-exceptional bin. On the other hand, 
a non-exceptional bin will have roughly $m/n$ samples 
with standard deviation $\sqrt{m/n}$. When $m/n \gg m \eps/\sqrt{kn}$  
(which happens in the regime $\eps \ll \sqrt{k/n}$), 
in order for the exceptional bins to be distinguishable, 
we would need that $m  \eps/\sqrt{kn} \gg \sqrt{m/n}$ or $m \gg k/\eps^2$ many samples. 
Using a careful information-theoretic argument, we formalize this intuition 
to show that $\widetilde \Omega(k / \eps^2) $ is indeed a sample lower bound in this regime.

\paragraph{Sample-Efficient Tester.} 
The starting point of our efficient tester is the ``Testing-via-Learning'' 
approach of \cite{AcharyaDK15}. {Very roughly speaking}, 
we first design a learning procedure which outputs a distribution $\hat \p$ 
that would be close to $\p$ in $\chi^2$ {divergence, assuming that} $\p$ was in fact a $k$-histogram. 
Then we use a $\chi^2$/$\ell_1$ tolerant tester, {in the spirit of the one introduced in~\cite{AcharyaDK15},}
to distinguish between the cases that $\p$ is close to $\hat \p$ in $\chi^2$-divergence
versus far from $\hat \p$ in $\ell_1$-distance. {We emphasize that the latter 
step is significantly more challenging than this rough outline suggests, 
as it is unclear how to perform the first step exactly. 
Instead, we design a specific learning algorithm with an implicit ``hybrid'' learning guarantee
(see Lemma~\ref{lem:sieve}),
which in turns requires us to considerably generalize and adapt the ``tolerant testing part'' to avoid spurious discrepancies (introduced in the imperfect learning stage) which may lead to false negatives.}

To implement the first step,
we follow the general ``learn-and-sieve'' idea suggested in \cite{Canonne16}, 
with important modifications to address the flaw in their approach and its analysis.
In particular, suppose that $\p$ is a $k$-histogram. Then, if we \emph{knew} 
the corresponding $k$ intervals (that make up the partition for the $k$-histogram), 
it suffices to learn the mass of $\p$ on each interval, 
and let $\hat \p$ be uniform on each interval (with the appropriate total mass). 
{Of course, a key source of difficulty arises from the fact that} 
we do not know the partition {a priori}. 
{To circumvent this issue,}
we divide $[n]$ into (roughly) $K= \Theta(k)$ 
intervals and try to detect if $\p$ is far in $\chi^2$ divergence 
from being uniform on any of these intervals. 
If an interval from our partition incurs large $\chi^2$ error 
{(we call such an interval {\em bad})}, 
we know that $\p$ must not be constant within this interval. 
Therefore, we proceed to subdivide these bad intervals into roughly 
equal parts, and recurse on the $\Theta(k)$ intervals in our new partition. 
{Assuming $\p$ is a $k$-histogram,} 
we subdivide at most $k$ intervals in each iteration, 
since there could be at most $k$ intervals from any {interval} partition of $[n]$ 
where $\p$ is not constant. 
Hence, in each iteration, we decrease the mass of the bad intervals by at least a constant factor.
We repeat the process for at most $O(\log(1/\eps))$ many iterations;
{after this many iterations,} 
the total mass of the bad intervals will become $O(\eps)$, 
and thus they may be safely ignored.

A significant difference between our method and the approach from \cite{Canonne16} 
lies in the method of \emph{sieving}. In \cite{Canonne16}, 
it was only said that the algorithm would filter out a subset of \emph{breakpoint} intervals 
based on the $\chi^2$ statistics (see, e.g.,~\cite{AcharyaDK15}) with the goal of reducing discrepancy; 
this is where the main gap in their analysis lies, and the particular (flawed) approach 
they suggested does not seem to be fixable~\cite{CanonnePCcom}. 
On the contrary, we characterize the exact set of intervals that need to (and can) be removed with 
a new definition of \emph{bad} intervals with respect to a given partition $\I$ of $[n]$ 
(see Definition \ref{def:bad-interval}). 
Based on that, our approach is to search for \emph{any} subintervals $J$ 
{(not necessarily an interval in $\I$)} on which 
the $\chi^2$ {divergence} between $\p$ and $\hat \p$ ---
{an approximation of $\p$ assuming $\p$ is uniform over intervals within the given partition} ---
is more than $\widetilde \Omega(\eps^2/k)$. 
{For an interval $I$ from the partition $\I$, we show the inclusion 
of such a ``bad subinterval'' $J \subseteq I$ then certifies the ``badness'' of $I$ itself.}
To find such a $J$, we need a technique for accurately approximating $\p(J)$ 
simultaneously for all intervals $J \subseteq [n]$, 
{in both absolute and relative error. We note that this is a notion of approximation 
much stronger than what classical tools from empirical process theory, 
such as the VC inequality (see, e.g.,~\cite{DL:01}), provide.} 
Notice that for a \emph{fixed} interval $J \subseteq [n]$, 
taking the empirical distribution over $b$ samples gives an estimate $\q$ of $\p$ 
such that $\abs{\q(J)- \p(J)} < \sqrt{\p(J)/b}$   
with constant probability.
By taking $\Theta( \log(n) )$ batches of samples (each containing $b$ i.i.d.\ samples from $\p$), 
and computing the median value of all of the $\q(J)$'s, with high probability for each $J$, 
we then obtain an estimate $\pmap(J)$~\footnote{Notice that $\pmap$ is neither a distribution nor a measure, but just a map from intervals to positive real values.} for which the above condition holds.
Using this subroutine, 
as long as $b$ is at least $ \Omega( k/\eps^2)$, 
we can ensure that $|\pmap(J)-\p(J)|^2/\p(J) \ll \eps^2/k$, 
{and we can then safely use our estimate $\pmap(J)$ as a proxy for $\p(J)$ 
for the detection of those ``bad subintervals'' 
for which $|\p(J)-\hat \p(J)|^2 / \hat \p(J)$ is large, 
which in turns certify the bad intervals from a given partition. }
This suffices \emph{unless} $\p(J)$ is substantially larger than our estimate $\hat \p(J)$.

Unfortunately, 
the ratio between 
$|\pmap(J)-\p(J)|^2/\hat \p(J)$ and 
$|\pmap(J)-\p(J)|^2/\p(J)$ (in particular $\p(J) / \hat \p(J)$) can be unbounded 
when $\p(J)$ is smaller than $1/b$.
In such a case, in a collection of $b$ samples from $\p$, 
we are likely to see no samples {in} $J$, and thus our empirical estimate $\hat \p(J)$ will be $0$. 
We can fix {this issue (i.e., the case where $\hat \p(J)$ is actually $0$)}
by mixing both $\p$ and $\hat \p$ with the uniform distribution, 
thus allowing us to assume that $\hat \p(J) \geq |J|/2n \geq 1/(2n)$.
Yet, this still leaves a potential gap of roughly $n/b$ 
between the ratio of $\p(J)$ and $\hat \p(J)$.
Fortunately, if we {select} $b \gg \sqrt{nk}/\eps$, 
we will have that $|\pmap(J)- \p(J)|^2/\p(J) \ll \eps/\sqrt{nk}$, 
and even accounting for losing a factor of $b/n$, 
we will still have that $|\pmap(J)- \p(J)|^2/ \p(J) \ll \eps^2/k$. 
This implies that we will successfully detect any bad intervals and achieve our learning guarantees.

\subsection{Prior and Related Work} \label{ssec:prior}

The field of \emph{distribution property testing}~\cite{BFR+:00} 
has been extensively investigated in the past
couple of decades, see~\cite{Rub12, Gol:17, Canonne15-survey} 
for two recent surveys and a book on the topic.
A large body of the literature has focused on characterizing the sample size needed to test properties
of arbitrary discrete distributions of a given support size. This regime is fairly well understood 
for many properties of interest, and in particular for symmetric properties 
(i.e., invariant under permutation of the domain). 
For a range of such properties, there exist sample-optimal 
testers~\cite{Paninski:08, CDVV14, DKN:15, AcharyaDK15, 
DiakonikolasK16, DiakonikolasGPP16, CDS17, Gol:17, DGPP17, CDKS18}.
We note that the property of being a histogram is not symmetric.

Motivated by the question of building provably good succinct representations 
of a dataset from only a small subsample of the data,~\cite{ILR12} 
first introduced histogram testing as a preliminary, 
ultra-efficient decision subroutine to find the best parameter $k$ for the number of bins. 
They provided an algorithm for this task which required $\wt O(\sqrt{kn}/\eps^5)$ samples 
from the dataset, a sample complexity which beats the na\"ive approach 
(reading and processing the whole dataset) for small values of $k$ 
and relatively large values of the accuracy parameter $\eps$. 
Subsequent work~\cite{CanonneDGR18} reduced the dependence on $\eps$ 
from quintic to cubic, giving an algorithm with sample complexity $\wt O(\sqrt{kn}/\eps^3)$. 
This bound was, however, still quite far from the ``trivial'' lower bound of $\Omega(\sqrt{n}/\eps^2)$, 
which follows from a reduction to uniformity testing (i.e., the case $k=1$)~\cite{Paninski08}.

Prior to the current work, an $\wt \Omega ( \sqrt{n} / \eps^2 + k /\eps )$ sample complexity 
lower bound was obtained in~\cite{Canonne16}. The latter work also claimed a testing 
algorithm with sample complexity $ \wt O (  \sqrt{n} / \eps^2 + k / \eps^3 )$.
Unfortunately, as pointed out in~\cite{CanonnePCcom}, 
the algorithm proposed in~\cite{Canonne16} is incorrect due to a technical flaw in the analysis. 
This leaves the sample complexity of the problem open for even constant $\eps$.
The lower bound of~\cite{Canonne16} is valid
and relies on a reduction of histogram testing to the well-studied problem 
of support size estimation. Consequently, it provably cannot be improved to provide either
(i) a quadratic dependence on $\eps$, i.e., $\wt \Omega(k/\eps^2)$, 
or (ii) coupling between the two domain parameters $k,n$, i.e., $\wt \Omega(\sqrt{nk}/\eps)$. 
Our work resolves the complexity of histogram testing, 
for the entire parameter range, within logarithmic factors.

Finally, we note that a number of works have obtained algorithms 
and lower bounds for related, yet significantly different, testing problems. 
Specifically,~\cite{DiakonikolasK16} gave a sample-optimal testing algorithm 
for the special case of our problem where the $k$ intervals are known {\em a priori}. 
This special case turns out to be significantly simpler.
Moreover, a number of works~\cite{DKN:15, DKN:15:FOCS, DKN17} 
have obtained identity and equivalence testers {\em under the assumption} 
that the input distributions are $k$-histograms.

\subsection{Preliminaries} \label{ssec:prelims}
We denote by $\totalvardist{\p}{\q}$ the total variation (TV) distance 
between probability distributions $\p,\q$ over $[n]\eqdef \{1,2,\dots,n\}$, defined as
\[
	\totalvardist{\p}{\q} \eqdef \sup_{S\subseteq [n]} (\p(S)-\q(S)) = \frac{1}{2} \sum_{i=1}^n |\p(i)-\q(i)|,
\]
where $\p(S) \eqdef \sum_{i\in S}\p(i)$. 
We will make essential use of the $\chi^2$-divergence of $\p$ with respect to $\q$, defined as
\[
	\dchi{}{\p}{\q} \eqdef \sum_{i=1}^n \frac{\lp( \p_i - \q_i \rp)^2}{\q_i}.
\]
We will also require generalizations of these definitions on restrictions of the domain. 
In particular, given $S \subseteq [n]$, we use the notation
$\totalvardistrestr{S}{\p}{\q} \eqdef (1/2) \sum_{i \in S} |\p(i)-\q(i)|$
and $\dchi{S}{\p}{\q} \eqdef \sum_{i \in S} \lp( \p_i - \q_i \rp)^2/\q_i$.
We note that for any $S \subseteq [n]$, it holds that $ \TV^{S}(\p, \q)^2 \leq \frac{1}{4} \dchi{S}{\p}{\q}$.

The asymptotic notation $\tilde{O}$ (resp. $\tilde{\Omega}$) suppresses logarithmic factors in its argument, 
i.e., $\tilde{O}(f(n)) = O(f(n)\log^c f(n))$ and 
$\tilde{\Omega}(f(n)) = \Omega(f(n)/\log^c f(n))$, where $c>0$ is a universal constant. 
The notations $\ll$ and $\gg$ intuitively mean ``much less than'' and ``much greater than'' respectively. 
Formally, we write $f(n) \ll g(n)$ to denote that $f(n) < c \cdot g(n)$, for some universal constant 
$c>0$. 

\section{Main Algorithmic Result: Near-Optimal Histogram Tester} \label{sec:ub}
\paragraph{A preliminary simplification.} 
Without loss of generality, we will assume that $\p(i) \geq \frac{1}{2n}$ for every $i\in[n]$. 
Indeed, this can be ensured by mixing the unknown distribution 
with the uniform distribution $\mathbf{u}_n$ on $[n]$ beforehand, 
i.e., $\p' \eqdef \frac{1}{2}(\p+\mathbf{u}_n)$ 
(see Fact~\ref{fact:uniform-mix} in Appendix for how to sample from $\p'$ efficiently).
It is easy to see that $\p'$ remains a histogram after mixing: $\p'\in \mathcal H_k^n$ if $\p \in \mathcal H_k^n$, and $\p'$ is at least $(\eps/2)$-far away from every histogram if $\p$ is $\eps$-far from every histogram.

\paragraph{Testing via Learning.} The high-level approach is to leverage 
the \emph{Testing-via-Learning} framework proposed in \cite{AcharyaDK15}. 
In particular, suppose we have a learning algorithm capable of constructing a hypothesis 
$\hat \p$ that is close to $\p$ in $\chi^2$ divergence when $\p \in \mathcal H_n^k$. 
Then, (1)~if $\p\in\mathcal{H}_k^n$, we will have that $\p$ and $\hat{\p}$ 
are close \emph{and} (as a consequence of this) that $\hat{\p}$ is close to being a $k$-histogram. 
Yet, (2)~if $\p$ is far from being a $k$-histogram, then by the triangle inequality 
we must have either that $\hat{\p}$ is far from being a $k$-histogram, 
or that $\p$ and $\hat{\p}$ are far from each other in $\ell_1$ distance. 
We can use dynamic programming to efficiently check the explicit description 
is indeed close to a $k$-histogram in $\ell_1$ distance (see Lemma 4.11 of \cite{CanonneDGR18}). 
To verify that $\p$ and $\hat{\p}$ are close, we will use a result of~\cite{AcharyaDK15} on tolerant identity testing.
In particular, given an explicit description $\hat \p$, the tester takes samples 
from the unknown distribution $\p$ and decides whether $\p$ and $\hat \p$ 
are closed in $\chi^2$ divergence or far in $\ell_1$ distance.
We remark that $\hat \p$ can be relaxed to be a positive measure.

\begin{lemma}[Adapted from Lemmas 2 and 3 \cite{AcharyaDK15}] \label{prop:chi-stats}
Let $\p$ and $\hat \p$ be a distribution and a positive measure defined on $[n]$ respectively.
Fix $\eps \in (0, 1)$ and  let $\A = \{ i \in [n]: \hat{\p}(i) \geq \eps/(50n) \}$.
There exists a tester \algname{Tolerant-Identity-Test}, which takes $\Poi(m)$ i.i.d.\ samples 
from $\p$ and outputs \textsf{Accept} if $ \dchi{ \mathcal A }{\p}{\hat \p} \leq \eps^2/500 $ 
and \textsf{Reject} if $\TV^{\A}\lp( \p, \hat \p \rp) \geq \eps$ with constant probability.
\end{lemma}

\paragraph{Outline for Learning.} If $\p \in \mathcal H_k^n$ and 
we know the partition of $\p$ in advance, one can learn $\p$ up to error 
$\eps^2$ in $\chi^2$ divergence with $\Theta(k / \eps^2)$ samples 
(following the analysis of Laplace estimator from \cite{KamathOPS15}).
Without the partition information, we can nonetheless achieve a weaker guarantee. 
That is, we can output a fully specified measure $\hat \p$ on $[n]$, 
together with a subset of the domain, $\G \subseteq [n]$, 
such that $\dchi{\G}{\p}{\hat \p}$ is small. 
In particular, we can achieve the guarantee in three steps.
 \begin{enumerate}
  \item[(i)] \label{step:1} Equally divide the domain $[n]$ into $K \gg k$ many intervals (Lemma~\ref{lem:sub-divide}).
  \item[(ii)] Output a succinct measure $\hat{\p} $ that is constant on each interval specified by Step (i) (Section~\ref{ssec:learn}).
  \item[(iii)] Identify the intervals $\Bad$ where $\dchi{\Bad}{\p}{\hat \p}$ is large (Section~\ref{ssec:sieve}) 
  and take $\G = [n] \backslash \B$.
 \end{enumerate}
The fact that we only have $\p$ and $\hat \p$ close in $\chi^2$ divergence on a subdomain $\G$ 
is a reasonable compromise, as long as $\p(\B), \hat\p(\B) \ll \eps$: 
if $\p$ is $\eps$-far away from $\hat \p$ in $\ell_1$ distance on $[n]$, 
$\p$ is at least $(\eps -\p(\B) - \hat \p(\B))$-far away from $\hat \p$ on $[n] \backslash \B$. 
Otherwise, we may take more samples from $\p$ restricted to $\B$ 
and subdivide the problematic intervals identified in Step (iii). 
Repeating the above steps iteratively then brings us to the case $\p(\B) \ll \eps$.

\paragraph{Equitable Partition.} \label{ssec:partition}
The first step is to divide the domain into $\Theta(k)$ many intervals 
over which the masses of $\p$ are approximately equal. 
As shown in \cite{AcharyaDK15}, this can be done with $\widetilde \Theta(k)$ many samples 
through a routine we denote by \algname{Approx-Divide}. 
We also need a routine for subdividing a set of disjoint intervals into even lighter subintervals. 
Nonetheless, one can reduce the subdividing task to domain partitioning 
by running \algname{Approx-Divide} on the subdistribution restricted to the set of disjoint intervals. 
For clarity of exposition, the proofs for this refinement of the subroutine from~\cite{AcharyaDK15} 
are provided in Appendix~\ref{appendix:partition}.

\begin{lemma} \label{lem:sub-divide}
There exists an algorithm \algname{Approx-Sub-Divide} that, 
given parameters $\delta \in(0,1]$ and integer $B > 1$, 
as well as a set of disjoint intervals $\mathcal I = \{ I_1, I_2, \cdots, I_q  \}$, 
given sample access to $\p$ on $[n]$, outputs a list of partitions $\mathcal S_{1},\dots,\mathcal S_{q}$, 
where $\mathcal S_{i}$ is the partition of the interval $I_i \in \mathcal I$, 
such that the following holds with probability at least $1 - \delta$.
\begin{enumerate}
\item[(i)] The algorithm uses $O\lp( B / \p(\I) \cdot \log \big(B / \delta \big) \rp)$ samples.
\item[(ii)] The output contains at most $(8B + q)$ intervals in total.
\item[(iii)] Every non-singleton interval $S \in \bigcup_{j=1}^q \mathcal S_j$ satisfies $\p(S) \leq \p(\mathcal I) \cdot 16/B$. 
\end{enumerate}
\end{lemma}

\subsection{Simultaneously Estimating Mass of Intervals} \label{ssec:learn}
In this section, we first introduce \algname{Interval-Mass-Estimate}, 
a subroutine that can accurately approximate the mass of $\p(J)$ 
for all intervals $J \subseteq [n]$ simultaneously, 
and then show how we can use it to learn $\p$ (assuming $\p \in \mathcal H_n^k$).

\algname{Interval-Mass-Estimate} first divides the number of samples drawn 
into $\Theta( \log (n/ \delta) )$ batches. For an interval $I$, 
we compute the estimate (number of samples falling in $I$ divided by the batch size) 
for each batch separately and compute the median over the statistics. 
This is often referred as the ``Median Trick'' 
and is crucial in achieving the learning guarantees with high probability. %

\begin{lemma} \label{lem:interval-mass-estimate}
Let be $\p$ be supported on $[n]$ such that $\p(i) \geq 1/(2n)$. 
Fix $b \in \mathbb Z^{+}$ and $\delta\in(0,1]$.
The algorithm \textup{\algname{Interval-Mass-Estimate}} 
takes $6 b \log (n/\delta)$ i.i.d.\ samples from $\p$ and 
outputs $\pmap$, a map from sub-intervals of $[n]$ to real values, 
such that, with probability at least $1 - \delta$, for every sub-interval $I \subseteq [n]$ 
it holds that ${ \p(I) }/{  \pmap(I) } \leq   \max (2\, , 8  n/b)$, ${\pmap(I) }/{ \p(I) } \leq 3$ 
and $\abs{ \pmap(I) - \p(I)} \leq \sqrt{ \p(I) / b }$.
\end{lemma}
\begin{proof}[Proof of Lemma~\ref{lem:interval-mass-estimate}]
We provide below the pseudocode of the algorithm, before analyzing its guarantees.
\begin{algorithm}[H]
\caption{Interval-Mass-Estimate} \label{alg:interval-mass-estimate}
    \begin{algorithmic}[1]
    \Require $m$ i.i.d.\ samples from distribution $\p$ on $[n]$; failure probability $\delta$.
    \State Let $T \eqdef  6 \log (n / \delta)$.
    \State Split the samples into $T$ batches $S^{(1)}, S^{(2)}, \cdots,  S^{(T)}$, each of $b\eqdef m/T$ samples.
    \State Initialize a Map $\pmap\colon \mathcal I \mapsto \R$ which maps sub-intervals of $[n]$ to real values.
    \For{all intervals $I \subseteq [n]$}

    \State Let $X_j^{(i)}$ be the number of samples falling in interval $I_j$ from samples $S^{(i)}$.
    \State Compute the median $ X_j \eqdef \Median_{1\leq i\leq T}\big( X_j^{(i)} \big) $. Accordingly, set
    \[
    \pmap ( I ) \eqdef \max \!\lp( \frac{X_j}{b}, \frac{ \abs{I} }{2n} \rp).
    \]
    \EndFor
    \State Output the map $\pmap$.
    \end{algorithmic}
\end{algorithm}
\noindent The following claim about the ``Median Trick'' will be useful.
\begin{claim} \label{clm:interval-median}
Let $I$ be an interval (or subset) over $[n]$ and $\p$ a distribution over $[n]$. 
For $\delta \in (0,1)$, assume one takes $T \eqdef 3 \log (1 / \delta)$ batches 
of i.i.d.\ samples from $\p$ where each batch is of size $b \in \mathbb Z^{+}$. 
Denote by $X_I^{(i)}$ the number of samples falling in $I$ from the $i$-th batch, 
for $1\leq i\leq T$. Then, the median $X_I \eqdef \Median_i(  X_I^{(i)})$ satisfies
\begin{align*}
&\abs{ \frac{ X_I }{b} - \p(I) } \leq 3 \sqrt{ \frac{\p(I)}{b}  } \quad \text{ and } \quad
\frac{ X_I }{b}  \leq 3 \p(I) \,,
\end{align*}
with probability at least $1 - \delta$.
\end{claim}
\begin{proof}
For $X_I^{(1)}$, it follows from standard application of Chernoff Bound 
and Markov's Inequality that
\begin{align}
& \abs{ \frac{ X_I }{b} - \p(I) } \leq 3 \sqrt{ \frac{\p(I)}{b}  }  \label{eq:cheby-cond}\\
& \frac{ X_I }{b}  \leq 3 \p(I) \label{eq:markov-cond}
\end{align}
each individually with probability at least $2/3$.
Let $a_i$ be the indicator variable such that Equation~\eqref{eq:cheby-cond} 
holds for $X_I^{(i)}$ for $1 \leq i \leq T$. If more than half of $a_i$s are true, 
then it follows the median $X_I$ will also satisfy Equation~\eqref{eq:cheby-cond}. 
Since $a_i$s are independent, this happens with probability precisely $Y > T/2$, 
where $Y \sim Bin(T, 2/3)$. By the multiplicative Chernoff's Bound, 
the failure probability is bounded above by 
\[
\Pr \lp(Y \leq \frac{1}{2} T  \rp)
= \Pr \lp (Y \leq (1 - \frac{1}{4}) \cdot \frac{2}{3} \cdot T    \rp)
\leq 
\lp( \frac{ \exp( -1/4 ) }{ (3/4)^{3/4} }  \rp)^{ 2/3 \cdot T } \,.
\]
Note that we indeed have $\lp( \frac{ \exp( -1/4 ) }{ (3/4)^{3/4} }  \rp)^{ 2/3 \cdot T } \leq \delta/2$ 
given $T \geq 3 \log(1 / \delta)$. A similar argument holds for Equation~\eqref{eq:markov-cond} 
and the claim follows by applying the union bound.
\end{proof}

With this in hand, we are ready to establish Lemma~\ref{lem:interval-mass-estimate}. 
Let $I \subseteq [n]$ be an arbitrary interval and $X_I$ 
be the median of the numbers of samples falling in from each batch. 
By Claim~\ref{clm:interval-median}, we have that
\begin{align}
&\abs{\frac{X_I}{b} - \p(I)} \leq \sqrt{ \p(I)/b }\,,\\
&\frac{X_I}{b} \leq 3 \p(I) \label{eq:markov-like}
\end{align}
with probability at least $1 - \delta/(2n^2)$. 
Since $\p(i) > 1/(2n)$ for any $i \in [n]$, 
the max operation decreases the distance between $\p$ and $\hat \p$ pointwise. 
Hence, with probability at least $1 - \delta/(2n^2)$, we have
\begin{align} \label{eq:interval-learn-deviation}
\abs{\pmap(I)  - \p(I)}  \leq \sqrt{ \p(I)/b }\,.
\end{align}
Since there are at most $n^2$ intervals $I \subseteq [n]$ overall, 
by a union bound over all of them, we have that Equation~\ref{eq:interval-learn-deviation} 
holds simultaneously with probability at least $1 - \delta/2$ for every interval $I$.

To show that the ratio $\p(I) / \pmap(I)$ is bounded, we first consider the case 
$\p(I) \leq 4/b$. Since $\pmap(I) > {1}/{(2n)}$, 
it is easy to see that $\p(I) / \pmap(I) \leq 8 n/b$. Otherwise, we have $\p(I) > 4/b$. 
Following Equation~\ref{eq:interval-learn-deviation}, it holds that
\[
\abs{ \p(I) - \pmap(I) } \leq \sqrt{ \p(I) / b }
\leq \frac{1}{2} \p(I) \, ,
\]
This then implies that
\[
    \frac{\p(I)}{\pmap(I)}
    \leq 2\,.
\]
Lastly, since $\p(I) \geq \abs{I} / (2n)$, together with Equation~\ref{eq:markov-like}, we have
\[
\frac{\pmap(I)}{\p(I)} = \frac{ \max \lp( \abs{I} / (2n), X_I / b  \rp) }{ \p(I) } \leq 3.
\]
This concludes the proof.
\end{proof}

Let $\I$ be a partition of $[n]$. We try to learn $\p$ pretending that $\p$ 
is constant over each interval within $\I$ with the routine \algname{Empirical-Learning}. 
In particular, the algorithm uses \algname{Interval-Mass-Estimate} 
to obtain estimations of the mass of $I \in \I$ and then flattens 
the mass uniformly among elements $i \in I$. Notice that, due to the application of the median trick, 
the output is not necessarily a distribution but rather 
a positive measure\footnote{That is, $\hat{\p}$ might not sum to one, and thus is not itself a probability distribution.} 
$\hat{\p}$ on $[n]$ which is constant over each interval within $\I$.

\noindent If $\p$ is indeed a $k$-histogram, errors are only incurred on a special type of intervals 
(of which there are at most $k$), which we refer to as the \emph{breakpoint intervals}.

\begin{definition}[Breakpoint Intervals] \label{def:breakpoint}
Given a $k$-histogram $\p$ on $[n]$, we say that $i \in [n]$ is a \emph{breakpoint with respect to $\p$} 
if $\p(i) \neq \p(i+1)$; and that an interval $I\subseteq[n]$ is a
\emph{breakpoint interval (with respect to $\p$)} if $I$ contains at least one breakpoint.
\end{definition}

With Definition~\ref{def:breakpoint} in mind, we now specify
the formal learning guarantees. %
\begin{lemma} \label{lem:learn:modified}
Suppose $\p \in \mathcal H_k^n$. 
Let $\I$ be a partition of $[n]$ into $K$ intervals. Let $b \in \mathbb Z^{+}$,
$\delta\in(0,1]$ and $T \eqdef 6 \log (K / \delta)$.
There exists an algorithm \textup{\algname{Empirical-Learning}} that, 
given $m\eqdef Tb$ i.i.d.\ samples from $\p$, outputs a positive measure $\hat{\p}$ 
which satisfies the following with probability at least $1 - \delta$.
(i) $\hat \p$ is constant within each interval in $\I$. 
(ii) For every sub-intervals $J \subseteq I$ where $I\in\I$, 
is a \emph{non-breakpoint} interval with respect to $\p$, 
we have ${ \p(J) }/{  \hat \p(J) } \leq   \max (2, 8 n/b)$ and $\abs{\hat \p(J) - \p(J)} \leq \sqrt{ \p(J) / b }$.
\end{lemma}
\begin{proof}[Proof of Lemma~\ref{lem:learn:modified}] 
The pseudocode of the algorithm is provided below. 
\begin{algorithm}[H]
\caption{Empirical-Learning} \label{alg:learn}
    \begin{algorithmic}[1]
    \Require  $m$ i.i.d.\ samples from distribution $\p$ on [n]; 
    partition $\I = \{ I_1, I_2, \cdots, I_K \}$ of the domain~$[n]$; failure probability $\delta$.
    \State Let $\mathcal S$ denote the multiset of the $m$ i.i.d.\ samples.
    \State Construct the Interval Estimator $\pmap \gets \algname{Interval-Mass-Estimate}
    ( \mathcal S, \delta ).$
    \For{$1\leq j \leq K$}
    \State Compute the estimate
    $ \hat{\p}(i) \eqdef  \frac{ \pmap( I_j ) }{ \abs{I_j} } \qquad 1\leq j\leq K, i \in I_j \;.$
    \EndFor
    \State Output the measure $\hat{\p}$. 
    \end{algorithmic}
\end{algorithm}
\noindent Let $I$ be a non-breakpoint interval. 
It is easy to see that $\hat \p( I ) = \pmap (I)$.
By Lemma~\ref{lem:interval-mass-estimate}, 
it holds that
\begin{align*}
& \frac{ \p(I) }{  \hat \p(I) } \leq   \max (2 , 8 n/b) \, ,    \\
& \abs{\hat \p(I) - \p(I)} \leq \sqrt{ \p(I) / b } \,.    
\end{align*}
with probability at least $1 - \delta$. 
Conditioned on that, we easily have, for any sub-interval $J \subseteq I$, 
\[
\frac{ \p(J) }{  \hat \p(J) } = \frac{ \p(I) }{  \hat \p(I) }
 \leq   \max (2, 8 n/b) \;,
\]
since both $\p$ and $\hat \p$ are uniform within $I$.
Furthermore, we also have
\[
\abs{ \p(J) - \hat \p(J) }
= \abs{ \p(I) - \hat \p(I) } \cdot \frac{ |J| }{ |I| }
\leq \sqrt{ \frac{\p(I)}{ b} } \cdot \frac{ |J| }{ |I| }
= \sqrt{ \frac{\p(J)}{ b} } \cdot \sqrt{\frac{ |J| }{ |I| }}
\leq \sqrt{ \frac{\p(J)}{ b } }  \,,
\]
where the first equality follows from the fact that $\p$ and $\hat \p$ 
are both uniform within $I$, the second inequality follows from our conditioning, 
and the last inequality follows from $|J|\leq |I|$.
\end{proof}

By combining the two guarantees in (ii) in Lemma~\ref{lem:learn:modified}, 
one can see the $\chi^2$ divergence between $\p$ and $\hat \p$, 
restricted to the non-breakpoint intervals, will be at most $\eps^2$ with high probability 
if taking $\Theta(KT/\eps^2)$ many samples.
Following a result from \cite{KamathOPS15,Canonne16}, 
one only needs $\Theta(K/\eps^2)$ samples to learn a $K$-histogram up to $\eps^2$ error 
in this restricted notion of $\chi^2$ divergence. 
One may wonder whether this is enough for us, and if the stronger (but less natural) guarantees provided by~Lemma~\ref{lem:learn:modified}, which end up increasing the number of samples required, are necessary. 
As we will see in the next section, we indeed need not only that the $\chi^2$ divergence is small, 
but also that the ratio $\p(I)/\hat \p(I)$ is bounded for all non-breakpoint intervals. 
In particular, this latter property enables us to compute relatively accurate estimates of the $\chi^2$ divergence 
restricted to subintervals and (consequently) to tell whether $\p$ is constant 
or from far from being constant on an interval.

\subsection{Bad Interval Detection} \label{ssec:sieve}
While large contributions to the $\chi^2$ divergence 
(assuming the learning phase was successful) 
will only come from breakpoint intervals, not all of them will necessarily contribute significantly 
to the $\chi^2$ divergence. In particular, a breakpoint interval is only considered ``bad'', 
and needs to be filtered out, if the error incurred is proportional to the number of breakpoints within. 

We now give the formal definition of such a ``bad interval.''

\begin{definition}[$\eps$-Bad-Interval] \label{def:bad-interval}
Fix a partition $\mathcal I$ of $[n]$ containing $K$ intervals. 
Let $I \in \mathcal I$ be a breakpoint interval of $\p$. 
Furthermore, suppose $I$ contains $j-1$ breakpoints, 
i.e., $\p$ is $j$-piecewise uniform in $I$. 
We say that $I \in \I$ is an \emph{$\eps$-bad interval} with respect to $\hat \p$ and $\I$ if 
$\dchi{I}{\p}{\hat{\p}} \geq j \cdot \eps^2 / K$.
\end{definition}

The definition suits our purpose for two reasons.
(i) The total $\chi^2$ error between $\p$ and $\hat \p$ on the set of ``good'' intervals 
(complement of the set of ``bad'' intervals) is small. 
Indeed, let $\G \in \mathcal I$ be a set containing no $\eps$-bad intervals. 
Since there are at most $K$ intervals contained in $\G$ 
and $k$ breakpoints contained in the intervals in $\G$, it is easy to see that 
$\dchi{\G}{\p}{\hat{\p}} \leq  O( \eps^2 )$.
(ii) One can reliably separate bad intervals from non-breakpoint intervals assuming the learning phase was successful.
To see why, note that in that case every non-breakpoint interval $I$ 
satisfies $\dchi{J}{\p}{\hat{\p}} \ll \eps^2 / K$ for all $J \subseteq I$ with high probability.
On the contrary, for any bad interval $I$, we claim there must be a sub-interval 
$Q \subseteq I$ where $\dchi{Q}{\p}{\hat{\p}} \geq  \eps^2 / K$ 
and both $\p$ and $\hat \p$ are constant within. In particular, if $I$ is an $\eps$-bad interval 
that contains $(j-1)$ breakpoints, we then have a partition $\{Q_1, \cdots, Q_j\}$ 
of $I$ over which $\p$ is piecewise constant and at least one of them will have $\chi^2$ error at least $\eps^2 / K$.

Our next step is to show how we can leverage the separating condition 
to design an efficient bad interval detection mechanism. 
This is where our method \emph{significantly differs} from \cite{Canonne16}.
At a high level, we take another set of independent samples 
to get an estimate $\pmap(Q)$ of $\p(Q)$ for all $Q \subseteq [n]$ simultaneously. 
Then, we compare $\pmap(Q)$ with $\hat \p(Q)$ to see whether we have 
$\dchi{Q}{\p}{\hat{\p}} \geq  \eps^2 / K$, which would in turn 
imply the interval $I \supseteq Q$ from the given partition is $\eps$-bad.
We next provide the pseudo-code for \algname{Learn-and-Sieve}, 
which finds a positive measure $\hat \p$ on $[n]$ and a domain $\mathcal B$ 
such that $\dchi{[n] \backslash \mathcal B}{\p}{\hat \p} \leq O(\eps^2)$ 
provided $\p \in \mathcal H_k^n$. For the sake of exposition, we defer its 
detailed analysis to Appendix~\ref{appendix:learn-and-sieve}, and provide here an outline of the argument.
\begin{algorithm}
\caption{Learn-And-Sieve} \label{alg:raw_sieve}
\begin{algorithmic}[1]
\Require Sample access to $\p$; a partition $\I$ of $[n]$ containing $K$ intervals; 
accuracy $\eps$; failure probability $\delta$.
\State Let $m = C \cdot (   K/\eps^2 + \sqrt{Kn} / \eps ) \cdot \log(n / \delta) $ for a sufficiently large constant $C$.
\State Draw $2m$ i.i.d. samples from $\p$ and split the samples evenly into $\mathcal S_1, \mathcal S_2$.
\State $\hat{\p} \gets \algname{Empirical-Learning}\lp( \mathcal S_1, \I, \delta/4\rp)$, $\pmap \gets \algname{Interval-Mass-Estimate}(\mathcal S_2, \delta/4)$, $\Bad \gets \{\}$.
\For{all intervals $Q \subseteq I$ for some $I \in \mathcal I$}
    \If{
    $\pmap(Q)/{\hat{\p}(Q)} > 6 \cdot \max(1, \eps \sqrt{n/K})$ \textbf{or} $ \abs{ \pmap(Q) - \hat{\p}(Q) } > \frac{1}{2} \sqrt{ \hat{\p}(Q) \eps^2/K }$} \label{rej:1}
    \State Add $I$ to $\Bad$.
    \EndIf
\EndFor
\State Output \textsf{Reject} if $\Bad$ contains more than $k$ intervals. 
Otherwise, \Return $\Bad$, $\hat{\p}$. \label{line:more-than-k}
\end{algorithmic}
\end{algorithm}

\begin{lemma}[Sieving Lemma] \label{lem:sieve}
Given a partition $\I$ containing $K$ intervals, sample access to $\p$ on $[n]$ and $\delta \in (0,1)$. 
Then, the output of \textup{\algname{Learn-and-Sieve}} (Algorithm~\ref{alg:raw_sieve}) satisfies the following.
    (i)~Suppose $\p \in \mathcal H_k^n$. Then the algorithm returns a positive measure $\hat{\p}$ and $\Bad$ 
    such that $ \dchi{ [n] \backslash \Bad }{\p}{\hat{\p}} \leq \eps^2$ 
    with probability at least $1 - \delta$.
    (ii) The output $\mathcal B$ contains at most $k$ intervals (if the algorithm does not reject).
(iii) At most $O(( {K}/{\eps^2} + { \sqrt{Kn} }/{\eps} ) \cdot \log(n/\delta)  )$ samples are used.
\end{lemma}
\begin{proof}[Proof Sketch]
We claim that if $\p \in \mathcal H_k^n$, then $\Bad$ contains all the $\eps$-bad intervals 
and no non-breakpoint intervals with high probability. 
Let $I$ be a non-breakpoint interval. 
For $b = \Theta( m / \log(n/\delta) ) = \Theta(K/\eps^2 + \sqrt{Kn} / \eps)$, with high probability we have that
$\abs{ \pmap(Q) - \p(Q)  } \leq \sqrt{ \p(Q)/b }$,
$
\abs{ \hat{\p}(Q) - \p(Q)  } \leq \sqrt{ \p(Q)/b } 
$ 
and 
${\p(Q)}/{ \hat{\p}(Q)} \leq  \max( 2 \, , 8 \cdot n / b)$,
which follow from Lemmas~\ref{lem:interval-mass-estimate}
and~\ref{lem:learn:modified}. Combining this with the triangle inequality 
and our choice of $b$ implies the second condition of Line~\ref{rej:1} will be false. 
The first condition can be shown to be false by rewriting 
$\pmap(Q) / \hat \p(Q)$ as ${\pmap(Q)}/{ \p(Q) } \cdot {\p(Q)}/{ \hat \p(Q)}$, 
which are themselves bounded, with high probability, 
by $3$ and $\Theta(1) \cdot \max(1, \eps \sqrt{n/K})$ 
again by Lemmas~\ref{lem:interval-mass-estimate} 
and~\ref{lem:learn:modified} and our choice of~$b$.

Let $I$ be a breakpoint interval. We then have
$
\abs{\p(Q) - \hat{\p}(Q)} \geq \sqrt{ \hat \p(Q) \cdot {\eps^2}/{K} }
$ for some sub-interval $Q \subset I$. 
If $\p(Q)$ is light ($\p(Q) \leq 2 \eps / \sqrt{Kn}$), 
we can show that $\p(Q)/b \leq 1/4 \cdot \hat \p(Q) \cdot \eps^2/K$, 
making $\pmap(Q)$, our estimation for $\p(Q)$, sufficiently accurate 
such that the second condition of Line ~\ref{rej:1} will be true. 
Otherwise, as $b \gg \sqrt{Kn}/ \eps$, the estimation $\pmap(Q)$ 
will be within multiplicative factors of $\p(Q)$. If $\hat \p(Q)$ is not 
much lighter than $\p(Q)$, we can again show that 
$\p(Q)/b \leq 1/4 \cdot \hat \p(Q) \cdot \eps^2/K$. 
Otherwise, the first condition of Line~\ref{rej:1} will be true. 
Conditioned on the event that $\B$ includes all $\eps$-bad intervals 
and no non-breakpoint intervals, it is easy to see that $\Bad$ 
will contain no more than $k$ intervals and that 
$\dchi{ \mathcal I  \backslash \mathcal B }{\p}{\hat{\p}} \leq O(\eps^2)$. 
We note that points (i) and (iii) follow from the definition of the algorithm.
\end{proof}

\algname{Learn-and-Sieve} (Algorithm~\ref{alg:raw_sieve}) 
outputs a fully specified description $\hat \p$ and a subdomain $\G \eqdef [n] \backslash \B$ 
such that $\dchi{\G}{\p}{\hat \p}$ is small given $\p \in \mathcal H_n^k$. 
For testing purposes, this is a reasonable divergence from the ideal guarantee 
that $\dchi{}{\p}{\hat \p}$ is small \emph{as long as $\p(\B)$ is also small}. 
If so, we can set $\hat \p(i) = 0$ for $i \in \B$ and invoke \algname{Tolerant-Identity-Test} 
with $\p$ and $\hat \p$. If the test passes, we then know that 
$\TV^{\G} \lp( \p, \hat \p\rp) \leq \eps/2$: this together with $\p(\B) \leq \eps/2$ 
then gives $\TV \lp( \p, \hat \p \rp) \leq \eps$.

Unfortunately, running \algname{Learn-and-Sieve} only once 
we may have $\p(\B)=\Omega(1)$. To handle this, we will need more fine-grained sieving procedure, 
which uses \algname{Approx-Sub-Divide} to further partition the bad intervals detected 
and invokes \algname{Learn-and-Sieve} \emph{iteratively}.
In each iteration, the total mass of the bad intervals shrinks by a constant factor, 
allowing us to reach $\p(\B) \ll \eps$ in at most $O(\log(1/\eps))$ iterations. 
Doing so leads to our final algorithm 
whose pseudo-code (Algorithm~\ref{alg:main}) 
and detailed analysis are provided next. %

\subsection{Main Testing Algorithm} \label{sec:main-test}
We now provide our final testing algorithm, Algorithm~\ref{alg:main}, whose analysis leads to the upper bound stated in Theorem~\ref{thm:main}.
\begin{algorithm}[htp!]
  \caption{Divide-And-Learn-And-Sieve} \label{alg:main}
  \begin{algorithmic}[1]
    \Require Sample access to the distribution $\p$; domain size $[n]$; accuracy $\eps$.
    \State Set $\I^{(0)} = \B^{(0)} = \{ [n] \} $. 
    \State Set $T \gets 3  \log (1/\eps)$, $\delta \gets \frac{1}{100T}$, $t \gets 0$, $r = 1$.
    \While{$r > \eps/8$}
        \State $\S^{(t+1)} \gets \algname{Approx-Sub-Divide} \lp(32 k, \B^{(t)}, \delta \rp)$. \label{line:divide}
        \State Set $\I^{(t+1)} = (\I^{(t)} \setminus \B^{(t)}) \cup \S^{(t+1)}$.
        \Comment{Note that $\I^{(t+1)}$ is still a partition of $[n]$.}
        \State 
        $\Q^{(t+1)}, \hat{\p}^{(t+1)} \gets \algname{Learn-And-Sieve}\!\lp( \I^{(t+1)}, \frac{\eps}{ 4 \sqrt{T} }, \delta \rp)$. \label{line:learn-and-sieve}
        \State \textsf{Reject} if \algname{Learn-And-Sieve} outputs \textsf{Reject}.
        \State Set $\B^{(t+1)} = \S^{(t+1)} \cap \Q^{(t+1)}$.
        
        \State Take $ \ell = \Theta \lp( \log(1 / \delta) / \eps \rp)$ i.i.d.\ samples from $\p$.
        \State Let $X$ be the number of samples falling in $\B^{(t+1)}$. Set $r \gets X/\ell$. \label{line:mass-estimate}
        \State $t \gets t+1$
    \EndWhile
    \State  Denote $\G^{(j)} = \S^{(j)} \setminus \B^{(j)}$ for all $j\geq 1$.~\label{step:G-def}~\Comment{Note that the union of $\G^{(1)}, \cdots \G^{(t)}, \B^{(t)}$ forms a partition of the domain $[n]$.}
    \State We will consider the measure $\bar \p$ such that on intervals $I \in \G^{(j)}$, $\bar \p(i) = \hat{\p}^{(j)} (i)$; and on intervals $I \in \B^{(t)}$, $\bar \p(i) = 0$.
    \State\label{step:check1} Use Dynamic Programming to check whether there is a $k$-histogram that is $\eps/2$-close to $\bar \p$ in $\ell_1$ distance. If not, \textsf{Reject}. \Comment{Can be done in time $\poly(k, 1/\eps, n)$ as shown in~\cite[Lemma~4.11]{CanonneDGR18}.}
    \State\label{step:check2} Output \textsf{Accept} if $\algname{Tolerant-Identity-Test}\lp(\p, \hat \p, \eps\rp)$ outputs \textsf{Accept}. Otherwise \textsf{Reject}. \label{line:test}
  \end{algorithmic}
\end{algorithm}

\begin{theorem}[Upper Bound of Theorem~\ref{thm:main}, restated] \label{thm:main:ub}
There exists a testing algorithm for the class of $k$-histograms on $[n]$ 
with sample complexity $m = \wt O( \sqrt{nk}  / \eps + k/\eps^2 + \sqrt{n}/\eps^2)$ and running time $\poly(m)$.
\end{theorem}
\begin{proof}[Proof of Theorem~\ref{thm:main:ub}]
We first argue that the algorithm terminates in $T=O(\log(1/\eps))$ rounds with high probability. 
By Lemma~\ref{lem:sub-divide}, we have that 
$\p(I) \leq \frac{16}{32 k} \cdot \p(\mathcal B^{(t)}) = \frac{1}{2k} \cdot \p(\mathcal B^{(t)})$ 
for every non-singleton interval $I \in \mathcal S^{(t+1)}$ with probability at least $1 - \delta$. 
By Lemma \ref{lem:sieve}, the subroutine \algname{Learn-And-Sieve} selects (removes) at most $k$ intervals 
if it does not output reject.
Notice that $\Bad^{(t+1)}$ will not include any singleton as singletons cannot be breakpoint intervals.
Then, it holds that $\p(\mathcal B^{(t+1)}) \leq \p(\mathcal Q^{(t+1)}) \leq \frac{1}{2} \p(\mathcal B^{(t)})$.
Hence, the mass of $\mathcal B^{(t)}$ will drop below $\eps/100$ 
after at most $T := 3 \cdot \log(1/\eps)$ iterations with probability at least 
\begin{align} \label{eq:prob-calc}
(1 - \delta)^{T} = 
(1 - \frac{1}{100T})^{T} \geq \frac{9}{10} \, ,    
\end{align}
where the second inequality holds when $T$ is sufficiently large. 
On the other hand, at the end of the iterations, we have that 
$ r \leq  10 \cdot \p( \mathcal B^{(t)} ) \leq \eps/10$ with probability at least $\frac{9}{10}$ 
by Markov's inequality. Combining the two facts then gives that 
the algorithm exits the while loop in at most $T \eqdef 3 \cdot \log(1/\eps)$ 
iterations with probability at least $0.9\cdot 0.9 > 8/10$.

Furthermore, we claim that with probability at least $9/10$, 
it holds that $\p(\mathcal B^{(t)}) \leq \eps/4$ when the algorithm exits the while loop.
Suppose at the $j$-th iteration, we have $\p(\mathcal B^{(j)}) > \eps/4$. 
Then, by the multiplicative Chernoff bound, we have 
$$
\Pr \lp[ r \leq \eps/8  \rp]
\leq 
\Pr \lp[ r \leq  1/2 \cdot \p(\mathcal B^{(j)}) \rp]
\leq   \delta \;.
$$
Hence, following the same calculation as Equation~\ref{eq:prob-calc}, 
our claim holds with probability at least $9/10$.
Conditioning on (i), the algorithm terminates in $T$ iterations and (ii) 
\begin{align} \label{eq:small-bad-mass}
\p(\mathcal B^{(t)}) \leq \eps/4
\end{align}
when the algorithm exits the loop. 
We now proceed to argue it outputs the correct testing result with  probability at least $2/3$.

\paragraph{Completeness.} Suppose we have $\p \in \mathcal{H}_k^n$. 
At the $t$-th iteration, we claim that, with probability at least $1 - \delta$, it holds
\begin{align} \label{eq:inductive}
\dchi{\G^{(t)}}{ \p }{ \hat{\p}^{(t)} } \leq \frac{\eps^2}{16T} \,,
\end{align}
where $\G^{(t)} = \S^{(t)} \bs \B^{(t)}$ as defined in line~\ref{step:G-def} 
and $\hat{\p}^{(t)}$ is the learned distribution in the $t$-th iteration.
By Lemma \ref{lem:sieve}, it holds 
\[
\dchi{\I^{(t)} \bs \Q^{(t)}}{ \p }{ \hat{\p}^{(t)}} \leq \frac{\eps^2}{16T} \;.
\]
Since $\G^{(t)} = \S^{(t)} \bs \B^{(t)}$ is a subset of $\I^{(t)} \bs \Q^{(t)}$, 
the claim in Equation~\ref{eq:inductive} follows.
Recall that we condition on the algorithm running for at most $T$ iterations.
Combining this with Equation~\ref{eq:inductive}, if we denote $\G = \bigcup_{t=1}^T \G^{t}$, it holds 
\begin{align} \label{eq:chi-squre-bound}
\dchi{ \G }{ \p }{ \bar \p } \leq 
\sum_{t=1}^T \dchi{ \G^{(t)} }{ \p }{ \hat{\p}^{(t)} } \leq \eps^2/16
\end{align}
with probability at least $(1 - \delta)^{T} \geq \frac{9}{10}$.
Observe that $\G$ is precisely the sub-domain $\mathcal A = \{ i \in [n]: \bar \p(i) \geq \eps/(50n) \}$ 
that will be used to compute the $\chi^2$ statistic, 
since $\bar \p(i) \geq \frac{1}{2n}$ for $i \in \G$ and $\bar \p(i) = 0$ for $i \not \in \G$. 
Conditioning on Equation~\ref{eq:chi-squre-bound}, by Proposition \ref{prop:chi-stats}, 
Line~\ref{step:check2} will output accept with probability at least $9/10$ by Chebyshev's Inequality.
Then, Equation~\ref{eq:chi-squre-bound} together with the conditioning $\p(\mathcal B^{(t)}) \leq \eps/4$ 
also implies that $\totalvardist{ \p }{ \bar \p } \leq \eps/2$ and 
$\p \in \mathcal H_k^n$, line~\ref{step:check1} will also pass. 
Overall, the algorithm accepts with probability at least $2/3$.

\paragraph{Soundness.} Suppose now that $ \totalvardist{\p}{\h} > \eps$ for every $\h \in \mathcal H_k^n$.
For the sake of contradiction, assume that line~\ref{step:check1} and line~\ref{step:check2} both pass. 
By Line~\ref{step:check2} and the contrapositive of Proposition~\ref{prop:chi-stats}, 
we have that $\totalvardistrestr{\G}{\p}{\bar \p} \leq \eps/4$ with probability at least $\frac{9}{10}$.

By definition, we have that $\bar \p( \mathcal B^{(t)} ) = 0$. 
Since $\p(\mathcal B^{(t)}) \leq \eps/4$ by our conditioning, 
it then holds that $\totalvardist{\p}{\bar \p} \leq \eps/2$.
Then, by Line~\ref{step:check1}, there exists a $k$-histogram $\mathbf h^\ast$ 
satisfying $\totalvardist{\mathbf h}{\bar \p} \leq \eps/2$. 
By the triangle inequality, we have that $\totalvardist{\mathbf h^\ast}{\p} \leq \eps$. 
This contradicts the assumption that $\p$ is $\eps$-far from any $k$-histogram. 
Hence, at least one of the two lines will output reject with probability at least $2/3$.

\paragraph{Sample complexity.} Finally, the samples from the unknown distribution $\p$ 
are used in five different types of routines -- dividing (Line~\ref{line:divide}), 
learning, sieving (Line~\ref{line:learn-and-sieve}), testing (Line~\ref{line:test}), 
and mass estimation (Line~\ref{line:mass-estimate}). 
By Lemma \ref{lem:sub-divide}, in one iteration, the dividing phase uses 
$O\big ( k \log(k/\delta)   /  \p( \mathcal B^{(t)}) \big)$
samples, where $\p(\mathcal B^{(t)})$ is the mass of the to-be-divided intervals at the $t$-th iteration. 
Since $\p(\mathcal B^{(t)})$ shrinks exponentially in every iteration, 
the samples consumed are dominated by the last iteration. Hence, at most 
$
O\lp( k \log(kT) / \eps \rp)
$
samples are used in total in the dividing phase.

At the $t$-th iteration, the partition size $K \eqdef \abs{\mathcal I^{(t+1)}}$ 
is upper bounded by $O(T \cdot k )$. Hence, by Lemma \ref{lem:sieve}, 
the \algname{Learn-And-Sieve} procedure consumes in total
\begin{equation}
 O\lp( 
    \lp( \frac{T^2 k}{\eps^2} + \frac{ T\sqrt{T k n} }{ \eps} \rp) \cdot \log \lp( n T )  \rp)
\rp)
\end{equation}
samples. 
The process of testing the mass of $\mathcal B^{(t)}$ 
takes $\Theta( T \cdot \log(1/T) /\eps)$ samples in total.

After the algorithm exits the for loop, the chi-squared tester uses 
$\Theta( \sqrt{n} / \eps^2  )$ samples.
Thus, overall, the algorithm takes
\begin{equation}
O\lp( 
    \frac{\sqrt{n}}{\eps^2} +
    \lp( \frac{k}{\eps^2} \cdot \log^2(1 / \eps)  + \frac{ \sqrt{k n} }{ \eps}  \cdot  \log^{3/2} (1 / \eps) \rp) \cdot \log \lp( n \log(1 / \eps) \rp)
\rp)
\end{equation}
samples, where we summed what is used by the different routines 
and substituted $T = 3 \log(1 / \eps)$. This concludes the proof of the theorem.
\end{proof}

\section{Sample Complexity Lower Bound} \label{sec:lb}
In this section, we describe the hard instances of the histogram testing problem, 
which leads to our $\widetilde \Omega( \sqrt{kn} / \eps  + k / \eps^2 )$ sample complexity lower bound.
As is standard, we will apply the so-called Poissonization trick: 
we will relax $P$, the unknown object being tested, to be a positive measure with total mass $\Theta(1)$. 
We refer to such a measure as an \emph{approximate probability vector},
 and give the corresponding notion of histogram.

\begin{definition}[Approximate Probability Vector]
For $0< \nu <1$, we define the set of \emph{$\nu$-approximate probability vectors} (APV) on the domain $[n]$ by
$
\tilde{\mathcal P}^n(\nu)\eqdef 
 \{ 
     P : P_{i}  \in [0,\infty)  \, \forall i \in [n]\, ,
     \abs{ \norm{P}_1 - 1 } \leq \nu
 \}.
$ 
Accordingly, the set of \emph{histogram APV} is given by 
\[
\tilde{\mathcal H}^n_k(\nu) \eqdef
 \left\{ 
     P \in \tilde{\mathcal P}^n(\nu): \frac{P}{ \snorm{1}{P} } \in \mathcal{H}^n_k 
\right\}.
\]
\end{definition}
To establish our sample complexity lower bounds, 
instead of the multinomial model (where exactly $\ns$ samples are taken from a distribution $P$), 
we will instead work under the related \emph{Poisson sampling model}, 
where the number of samples is itself a Poisson random variable. 
Under this setting, given an unknown $P \in \tilde {\mathcal P}^n(\nu)$, 
the goal it to decide whether $P \in \tilde{\mathcal H}_k^n(\nu)$ or 
$P$ is at least $\eps(1 + \nu)$-far\footnote{The extra $(1 + \nu)$ factor accommodates 
the fact that $P$ may not be a distribution, i.e., $1 \leq \snorm{1}{P} < (1 + \nu)$.} 
from any $P' \in \tilde{\mathcal{H}}^n_k(\nu)$ in $\ell_1$-distance 
when given the vector $\{ M_1, M_2, \cdots M_n \}$, 
where $M_i \sim \Poi( m \cdot P_i )$. We denote the sample complexity 
of the problem by $\ns_{\rm{}hist}^{\rm{poi}}(n, k, \eps, \nu)$ 
and provide its formal definition below.

\begin{definition}[Histogram Testing under Poisson Model]
For $0 < \nu < 1$, $0 < \eps < 1$, define the sample complexity of histogram testing (under the Poisson model) as
\begin{align*}
&\ns_{\rm{}hist}^{\rm{poi}}(n, k, \eps, \nu) \eqdef
\min \bigg \{ 
\ns \geq 0 : \exists \hat I_m  \text{ s.t. } 
\sup_{P \in \mathcal{M}_{\eps}(\nu) } \Pr_{P} [ \hat I_m = 1 ] +  
\sup_{P \in \tilde{\mathcal{H}}^n_k(\nu) } \Pr_P [ \hat I_m = 0 ]
\leq \frac{1}{10} \bigg \} \,,
\end{align*}
where $\hat I_m$ is a binary indicator measurable with respect to 
$M = (M_1, M_2, \dots, M_n) \sim \poisson{\ns P_1}\otimes\poisson{\ns P_2}\otimes \cdots\otimes \poisson{\ns P_n}$,\footnote{We remark that the choice of the constant $1/10$ for error rate is arbitrary 
and our argument can easily be adapted to show a lower bound on the sample complexity 
under any constant error rate $\delta \in (0, 1/2)$.}, 
and $\mathcal{M}_{\eps}(\nu)$ is given by
\[
	\mathcal{M}_{\eps}(\nu) \eqdef \{ P \in \tilde{\mathcal P}^n : 
	\inf_{Q\in \tilde{\mathcal{H}}^n_k(\nu)} \frac{1}{2}\ell_1( P, Q ) \geq \eps(1 + \nu)  \}\,,
\]
that is, is the set of approximate probability vectors 
far (in $\ell_1$ distance) from being histogram APVs. %
\end{definition}

The core of the argument then lies in showing that $\ns_{\rm{}hist}^{\rm{poi}}(n, k, \eps, \nu)$ 
is bounded below by $\widetilde \Omega ( \sqrt{nk} / \eps )$ and $\widetilde \Omega ( k / \eps^2 )$,  
where $0 < \nu < 1$ is a constant.
To do so, we follow the idea of \emph{moment matching} illustrated 
in \cite{Valiant11,valiant2013estimating,Wu2016}. In particular, one first constructs 
two discrete non-negative random variables $U, U'$ whose first few moments are identical. 
Moreover, $U$ and $U'$ will be designed to have different properties 
such that one can use i.i.d.\ copies of $U$ (and $U'$) to generate random measures 
that are histograms (and far from histograms respectively). 

Our construction of such a pair of random variables is based on \emph{Chebyshev's polynomials}, 
a standard tool in approximation theory and the parameter estimation literature.  
The two variables will be supported on the roots of the polynomial 
$p(x) = x \left(x - \frac{1}{n}\right)\left(x - \frac{2}{n}\right) T_d\lp( 1 -  \frac{\sqrt{kn}}{C \cdot \log^2 n} \cdot x  \rp)$, 
where $T_d(\cdot)$ is the \emph{Chebyshev's polynomial} (of the first kind) and $C$ is a sufficiently large constant.
More precisely, $U$ will be supported on roots $r$ where the derivatives $p'(r) < 0$, 
$U'$ will be on roots where $p'(r) > 0$, and the probabilities will be proportional to $p'(r)$ accordingly. 
Consequently, $U$ will most likely be $1/n$ (hence, useful for histogram construction) 
and $U'$ will most likely be $0$ or $2/n$, each with non-trivial probabilities 
(hence, appropriate for non-histogram construction). Moreover, 
they will have their maximums bounded by $\widetilde O({1}/{\sqrt{kn} })$, 
which is crucial to achieve the nearly optimal lower bounds. 
The detailed construction and analysis are provided in Appendix~\ref{appendix:var-construct}.
\begin{lemma} \label{lem:rv-pair}
Given positive integers $k, n$ where $k<n$,
there exists a pair of non-negative random variable $U,U'$ supported on $[0, 1)$ and absolute constants $c,c'>0$ satisfying
    (i) $\Pr \lp[ U \neq \frac{1}{n} \rp] \ll \frac{k}{n}$.
    (ii) $\Pr \lp[ U' = 0 \rp] > 1/3$ and $\Pr \lp[ U' = \frac{2}{n} \rp] > 1/3$.
    (iii) $U,U' \leq { c' \log^2n }/{\sqrt{kn} }$.
    (iv) $\E[U] = \E[U'] = \frac{1}{n} ( 1 + O(\sqrt{k/n}) )$.
    (v) $\E[U^t] = \E[U'^t]$ for $1 \leq t \leq c \cdot \log n $.
\end{lemma}
We then proceed to construct two families of Approximate Probability Vectors, 
one of which belongs to $\tilde{\mathcal H}^n_k$ and the other far from it 
using the random variables stated in Lemma~\ref{lem:rv-pair}.
To do so, we define 
$
H = \lp( {1}/{n}+\eps U^{(1)}, \cdots, {1}/{n}+\eps  U^{(n)} \rp)$,
$H' = \lp( {1}/{n}+\eps U'^{(1)}, \cdots, {1}/{n}+\eps  U'^{(n)} \rp) \,
$
where \smash{$U^{(i)}$}, \smash{$U'^{(i)}$} are $n$ i.i.d.\ copies of  $U$, $U'$ in Lemma~\ref{lem:rv-pair}.

We address the two regimes $\sqrt{k/n} \leq \eps \log^2 n$ and $\sqrt{k/n} \geq \eps \log^2 n$ separately. 
In the former case, the heaviest elements among $H$ and $H'$ are roughly $ \widetilde{\Theta}( \eps / \sqrt{kn}  )$. 
Hence, when the algorithm takes $\widetilde{o}( \sqrt{kn}/ { \eps } ) $ samples, 
it rarely sees any element appearing a large number of times. 
By the moment-matching property of $U$ and $U'$, the probabilities of seeing some elements appearing 
$t$ times for $t \leq \log n$ are almost identical under $H$ and $H'$, 
therefore making $H$ and $H'$ indistinguishable. 
In the latter case, we have $\eps U \ll \frac{1}{n}$, implying that no elements in the measures 
are significantly heavier than the rest. As a result, $H$ and $H'$ are both almost uniform 
except with a different number of ``bumps'' (elements that are slightly heavier). 
Subsequently, the algorithm needs more samples (about $\widetilde \Omega(k / \eps^2)$) 
to tell whether a certain element is heavier than the rest, 
leading to a phase transition in the sample complexity of the problem.
We remark that whether the term $\widetilde \Omega(k / \eps^2)$ or the term $\widetilde \Omega(\sqrt{nk} / \eps)$ 
dominates depends exactly on the relationship between $\sqrt{k/n}$ and $\eps$ (omitting polylogarithmic factors). 
Combining the two regimes then gives us the following lower bound:
\begin{proposition} \label{lem:lower-bound-1}
There exists a constant $\nu \in (0,1)$ such that for any sufficiently large $n$ and $\eps \in  (0, 1/10)$, it holds 
$
\ns_{\rm{}hist}^{\rm{poi}}(n, k, \eps, \nu) \geq \Omega (
\max( 
{\sqrt{kn}}/{ ( \eps \log n ) } \, ,
{k}/{ ( \eps^2 \log^3 n ) }
)
).
$
\end{proposition}

\begin{proof}%
Our goal is to argue that $H$ and $H'$, specified in Equation~\ref{eq:hard:instance} below, 
satisfy the following properties: 
(i) $H$ and $H'$ are positive measures with total mass $1 + \nu$ for some $\nu \in (0,1)$ 
with probability at least $99/100$; 
(ii) $H$ is a $k$-histogram with probability at least $99/100$ 
and $H'$ is $\Omega(\eps)$-far away from any $k$-histogram with probability at least $99/100$; 
and (iii) the distributions of the vectors $M = \{ M_1, M_2, \cdots M_n \}$ and 
$M' =\{M'_1, M'_2, \cdots M'_n \}$, where $M_i \sim \Poi( m \cdot H_i )$ and 
$M'_i \sim \Poi( m \cdot H'_i )$, are $1/4$-close to each other in TV distance, 
as long as $m = o\lp(  \max \lp(  \sqrt{kn} / \lp( \eps \cdot \log n \rp), k / \lp(\eps^2 \cdot \log^3 n \rp)  \rp)  \rp)$.
If all these properties are satisfied, the result follows by applying Le Cam's Lemma.
To achieve the above, we define the positive measures $H$, $H'$ as follows:
 \begin{equation} \label{eq:hard:instance}
 H = \bigg( \frac{1}{n}+\eps U^{(1)}, \cdots, \frac{1}{n}+\eps  U^{(n)} \bigg) \, \quad \,
 H' = \bigg( \frac{1}{n}+\eps U'^{(1)}, \cdots, \frac{1}{n}+\eps  U'^{(n)} \bigg) \,
 \end{equation}
 where $U^{(1)}, \cdots, U^{(n)}$ and $U'^{(1)}, \cdots U'^{(n)}$ are 
 $n$ i.i.d.\ copies of the random variables $U$, $U'$ defined in Lemma~\ref{lem:rv-pair}.
 
 \smallskip

We proceed to verify that each property in our goal is satisfied.
\begin{enumerate}
\item[(i)]  We first observe that the mass of $H$ is simply $1+\eps \sum_{i=1}^n U^{(i)}$. 
As stated in Lemma \ref{lem:rv-pair}, we have 
$\E \lp[ U\rp] = \frac{1}{n} \lp( 1 + O(\sqrt{k/n}) \rp)$.
Since $U^{(i)}$s are just i.i.d.\ copies of $U$, this further implies that
$\E \lp[  \sum_{i=1}^n U^{(i)} \rp] = 1 + O(\sqrt{k/n}) =O(1)$. 
Hence, by Markov's inequality, we have that $\sum_{i=1}^n U^{(i)} =  \Theta(1)$
with probability at least $99/100$. Similar arguments hold for $H'$. This shows claim (i).

\item[(ii)] Turning to (ii), recall that by construction of $U$ we have that $\Pr\lp[ U \neq \frac{1}{n} \rp] \ll k/n$. 
Hence, with probability $99/100$, there are at most $(k-1)/2$ entries in $H$ with mass other than $1/n$, 
which makes $H$ a $k$-histogram.To argue the second part, i.e., that $H'$ is far from any $k$-histogram, 
we first lower bound the number of adjacent pairs $(i, i+1)$ such that $H_i = 0, H_{i+1} = \frac{2}{n}$. 
We call such an adjacent pair a ``right border pair.''
For $U'$, we have that $\Pr\lp[ U'=0 \rp], \Pr\lp[ U'= \frac{2}{n} \rp] > 1/3$. 
Hence, in expectation, there are at least $\frac{1}{9} (n-1)$ such right border pairs. 
On the other hand, the variance of the number of right border pairs is at most $n$. 
Therefore, by Chebyshev's inequality (and assuming $n$ large enough), 
with probability $99/100$, there will be at least $ \frac{1}{9} (n-1) - 100 \sqrt{n}  \geq \frac{1}{10} n$ 
right border pairs for sufficiently large $n$. 
This implies that $H'$ is at least $\lp( \frac{n}{10} - k \rp) \cdot \frac{\eps}{n} = \Omega(\eps)$ far from any $k$-histogram. 
This concludes claim (ii).

\item[(iii)] Let $P_M$ and $P_{M'}$ be the distributions of the tuple of $m$ samples seen by the algorithm. 
By the subadditivity of the total variation distance, we have that
\begin{align} \label{eq:product-tv}
\TV\!\lp( P_M, P_M' \rp) \leq n \cdot
\TV \lp( \E_{U}\!\lp[ \Poi\mleft( \frac{m}{n} + \eps m U \mright) \rp], \E_{U'}\!\lp[ \Poi\mleft( \frac{m}{n} + \eps  m U' \mright) \rp] \rp) \;.
\end{align}
To handle the right-hand-side, we will discuss the regimes 
$\sqrt{k/n} / \log^2 n < \eps$ and $\sqrt{k/n} / \log^2 n > \eps$ separately.

\textbf{Regime I ($\sqrt{k/n} / \log^2 n < \eps$)}. 
Notice that the term $\sqrt{kn} / (\eps \cdot \log n)$ dominates. 
So it suffices to show that $\TV\lp( P_M, P_M' \rp)$ is bounded by $1/4$ 
when $m = o\lp( \sqrt{kn} / \lp(\eps \cdot \log n \rp) \rp)$.
We will use the following lemma about the distance between mixtures of Poisson distributions.
\begin{lemma}[{\cite[Lemma~4]{wu2019chebyshev}}] \label{lem:TV-mixture}
Let $V, V'$ be random variables taking values in $[0, \Lambda]$. If $\E\lp[ V^j \rp] = \E \lp[ V'^j \rp]$ for $1\leq j \leq L$, 
then
\[
\operatorname{TV}\!\lp( \E\lp[\Poi(V) \rp], \E \lp[ \Poi(V') \rp]  \rp) \leq \lp( \frac{e \Lambda }{2 L} \rp)^L \;.
\]
\end{lemma}
Now let $\Lambda \eqdef m \mleft( \frac{1}{n} +  c'\eps \frac{\log^2 n}{\sqrt{k n}} \mright)$ and 
$L \eqdef c\log n$, where $c,c'>0$ are as in Lemma~\ref{lem:rv-pair}. 
It is straightforward to check from Lemma~\ref{lem:rv-pair} 
that the random variables $V \eqdef \frac{m}{n} + \eps m U$, $V' \eqdef \frac{m}{n} + \eps m U'$ 
satisfy the assumptions of Lemma~\ref{lem:TV-mixture}. 
This implies the existence of a (small) absolute constant $c''>0$ 
such that, if $m \leq c'' \min\mleft(n\log n, \frac{\sqrt{kn}}{\eps\log n}\mright)$, then by~\eqref{eq:product-tv}
\begin{align} \label{eq:first-regime-result}
\TV\!\lp( P_M, P_M' \rp) \leq 
n \cdot
 \mleft( \frac{e}{2c} \cdot m \mleft( \frac{1}{n\log n}+ c' \frac{\eps \log n}{\sqrt{kn}} \mright) \mright)^{c\log n}
 \leq n \cdot \lp( \frac{1}{2} \rp)^{c \log n}  < \frac{1}{4} \,,
\end{align}
for sufficiently large constant $c$.
Since $H$ and $H'$ satisfy all the properties listed, by Le Cam's Lemma, 
no algorithm can distinguish between the two distributions with probability more than $3/4 + 2/100 \leq 9/10$ 
when 
\[
m = o\mleft(\min\mleft( n\log n, \frac{\sqrt{kn}}{\eps\log n}\mright)\mright)
\]
under the Poissonization model. 
Notice that $\min\mleft( n\log n, \frac{\sqrt{kn}}{\eps\log n}\mright)$ is exactly 
$\frac{\sqrt{kn}}{\eps\log n}$ under the assumption $\eps \cdot \log^2 n \geq \sqrt{k/n}$. 

\textbf{Regime II ($\eps < \sqrt{k/n} / \log^2 n$)}.
Now the term $k / \lp( \eps^2 \cdot \log^3 n\rp)$ dominates.
We will now use a result from~\cite{informaticists}, restated below:
\begin{theorem}[Theorem 4 from~\cite{informaticists}] \label{thm:informaticists}
  For any $\Lambda>0$ and random variables $X,X'$ supported on $[-\Lambda,\infty)$, we have
  \[
      \totalvardist{\E[\poisson{\Lambda+X}]}{\E[\poisson{\Lambda+X'}]}
      \leq \frac{1}{2} \mleft( \sum_{\ell=0}^\infty \frac{\abs{ \E[X^\ell] - \E[X'^\ell] }^2}{\ell!\Lambda^\ell} \mright)^{1/2}\,.
    \]
\end{theorem}
\noindent Recall that the first few moments of $U, U'$ are identical, i.e.,
$ \E[ U^t ] = \E [U'^t] $ for $1 \leq t \leq L \eqdef c \cdot \log n$. 
Let $X\eqdef \eps m U$, $X'\eqdef \eps m U'$, and $\Lambda \eqdef \frac{m}{n}$. 
Notice that we indeed have $\abs{X}, \abs{X'} \leq \Lambda$ under the assumption $\eps \ll \sqrt{k/n} / \log^2 n$, since
$$
\max ( \abs{X}, \abs{X'} ) =
\eps \cdot m \cdot \max ( \abs{U}, \abs{U'} )
\leq 
\eps \cdot m \cdot O\lp(  \frac{\log^2 n}{\sqrt{kn}} \rp)
\leq m/n = \Lambda \;.
$$
Applying Theorem \ref{thm:informaticists} then gives
\begin{align*}
  4 \cdot \TV &\lp( \E\!\lp[ \Poi\lp( \frac{m}{n} + \eps m U \rp) \rp]\, , \E\!\lp[ \Poi \lp( \frac{m}{n} + \eps m U' \rp) \rp] \rp)^2 \\
      &\leq \sum_{\ell=0}^\infty \lp(\eps m\rp)^{2\ell} \frac{\abs{ \expect{U^\ell} - \expect{U'^\ell} }^2}{\ell!(m/n)^\ell}  \\
      &= \sum_{\ell=L+1}^\infty \lp( \eps^2 m n \rp)^{\ell} \frac{\abs{ \expect{U^\ell} - \expect{U'^\ell} }^2}{\ell!} \tag{the first $L$ moments match} \\
      &\leq \sum_{\ell=L+1}^\infty \mleft( \eps^2 mn \mright)^{\ell} \frac{1}{\ell!} \lp( \frac{ c'\log^2 n }{ \sqrt{kn} } \rp) ^{2\ell} \tag{$|U|, |U'|\leq \frac{ c'\log^2 n}{\sqrt{nk}} $ (Lemma~\ref{lem:rv-pair})} \\
      &= \sum_{\ell=L+1}^\infty \lp( c'^2 \cdot \log^4 n \cdot \eps^2 / k  \cdot m    \rp)^{\ell} \frac{1}{\ell!} \;.
\end{align*}
Set for convenience $\kappa \eqdef  c'^2 \cdot \log^4 n \cdot \eps^2 / k  \cdot m $. 
Notice that $\kappa \ll L$ when $m = o\lp( \frac{k}{ \eps^2 \log^3 n } \rp) $. 
Denoting by $Y$ a $\poisson{\kappa}$ random variable,  this leads to
\begin{align*}
  & 4\TV \lp( \E \lp[ \Poi\lp( \frac{m}{n} + \eps m U \rp) \rp]\, , \E \lp[ \Poi \lp( \frac{m}{n} + \eps m U' \rp) \rp] \rp)^2 \\
      &\leq \sum_{\ell=L+1}^\infty \frac{\kappa^\ell}{\ell!} = e^{\kappa} \probaOf{Y \geq L+1 } \\
      &\leq e^{\kappa} e^{-\frac{(L+1-\kappa)^2}{2(L+1)}} 
      = e^{-\frac{1}{2}\mleft(L+1+\frac{\kappa^2}{L+1}\mright)} 
      \leq e^{-\frac{L}{2}} \;,
\end{align*}
where the second inequality is by standard Poisson concentration 
(see, e.g., the note \cite{canonne-poisson}) and $\E [Y] = \kappa \ll L$. 
This immediately gives 
\begin{align} \label{eq:U-dist2}
\TV \lp( \E \lp[ \Poi \lp( \frac{m}{n} + \eps m U \rp) \rp]\, , \E \lp[ \Poi \lp( \frac{m}{n} + \eps m U' \rp) \rp] \rp) \leq  \frac{1}{2} e^{-L/4} \;.
\end{align}
Combining Equations \eqref{eq:product-tv} and \eqref{eq:U-dist2} then yields 
\begin{align} \label{eq:second-regime-result}
\TV \lp( P_M, P_M' \rp)
\leq n \cdot e^{-L/4}  =n \cdot e^{-c \log n /4} < 1/4
\end{align}
for sufficiently large $c$. 
Equations~\eqref{eq:first-regime-result} and ~\eqref{eq:second-regime-result} together 
with our choices of regimes then conclude the proof of point (iii).
\end{enumerate}
By Le Cam's Lemma, no algorithm can distinguish between 
the two distributions as constructed in Equation~\ref{eq:hard:instance} with probability 
more than $3/4 + 2/100 < 9/10$, 
when $m = o\!\lp( \max \lp( \sqrt{nk} / \lp( \eps \cdot \log n \rp) , {k}/{ \lp( \eps^2 \log^3 n \rp) } \rp) \rp) $ 
under the Poissonized sampling model.
\end{proof}

As previously mentioned, we can easily translate our lower bound result in the Poissonized sampling model 
to the Multinomial (standard fixed-size) sampling model by a standard reduction. 
Combining it with the known $\Omega(\sqrt{n}/\eps^2)$ bound (see~\cite[Proposition 4.1]{Canonne16}) 
then concludes our lower bound argument, and establishes the lower bound stated in Theorem~\ref{thm:main}; details follow.

\begin{proof}[Proof of Lower Bound Part of Theorem~\ref{thm:main}]
By Proposition~\ref{lem:lower-bound-1}, it holds that
$\ns_{\rm{}hist}^{\rm{poi}} \geq \Omega(k/(\eps^2 \cdot \log^3 k)  + \sqrt{kn} / ( \eps \cdot  \log k))$. 
We proceed to argue for sample complexity lower bound under standard sampling.
Suppose we are given a tester for fixed sample size such that it succeeds with high probability 
as long as it is given more than $m^*$ samples.
Assume that we want to use it to test whether an unknown measure $P$ is a $k$-histogram 
under the Poisson sampling model with $\tilde \ns \sim \Poi(\ns_{\rm{}hist}^{\rm{poi}}(n, \ell, \eps) \snorm{1}{P})$ samples.
We can construct an estimator which invokes the fixed sample size tester whenever $\tilde \ns \geq m^*$ and outputs fail otherwise. 

By our lower bound result for the Poisson sampling model, 
the estimator fails with probability at least $1/10$. 
On the other hand, the estimator based on the fixed-sample tester succeeds 
with high probability whenever $\tilde \ns > m^*$. 
Together this implies that
$$
m^* 
\geq 
(1 - \nu)
\cdot \ns\cdot 
\lp( 
    1 -  O \lp(
            (1 - \nu)^{-1} \cdot 
            \ns^{-\frac{1}{2}} 
            \rp)
\rp) \, ,
$$
where $\ns \eqdef \ns_{\rm{}hist}^{\rm{poi}}(n, \ell, \eps, \nu)$.
Since $\nu < 1$, it then holds 
$$
m^* > \Omega(1) \cdot  \ns_{\rm{}hist}^{\rm{poi}}(n, \ell, \eps, \nu) \geq \Omega(k/(\eps^2 \cdot \log^3 k)  + \sqrt{kn} / ( \eps \cdot  \log k)).
$$
Finally, we remark that the standard lower bound construction and analysis 
for uniformity testing can be shown to still apply to testing $k$-histograms (see~\cite[Proposition 4.1]{Canonne16}). 
This shows that we also have $m^* \geq \Omega ( \sqrt{n} / \eps^2 )$, and concludes the proof.
\end{proof}

\printbibliography
\appendix

\section*{APPENDIX}

\section{Deferred Proofs from Section~\ref{sec:ub}} \label{app:ub}
We provide in this appendix the proofs of some technical lemmas, 
which were omitted from the main paper in the interest of space. %

\begin{fact}\label{fact:uniform-mix}
Let $\mathbf u$ be the uniform distribution and $\p$ be an arbitrary distribution among $[n]$. Then, given sample access to $\p$, one can efficiently sample from $\p' \eqdef \frac{1}{2} (\p + \mathbf u)$.
\end{fact}
\begin{proof}
 Given a sample from $\p$, we generate a sample from $\p' \eqdef \frac{1}{2}(\p+\mathbf{u}_n)$ by outputting that sample with probability $1/2$, and a uniformly random value in $[n]$ otherwise. It is immediate to see this allows one to get $m$ i.i.d.\ samples from $\p'$ given $m$ i.i.d.\ samples from $\p$, and does not require knowledge of $\p$.
\end{proof}

\subsection{Proof of Equitable Partition: Lemma~\ref{lem:sub-divide}}
\label{appendix:partition}
To show the guarantee of \algname{Approx-Divide} in Lemma~\ref{lem:sub-divide}, we will rely on Claim~1 from \cite{AcharyaDK15}. We provide its proof for completeness.
\begin{algorithm}[H]
\caption{Approx-Divide} \label{alg:approx-divide}
    \begin{algorithmic}[1]
    \Require $m$ i.i.d. samples from distribution $\p$ on $[n]$.
    \State Let $\hat \p$ be the empirical distribution with respect to the $m$ samples.
    \State Denote by $i_1<\dots< i_T$ the elements $i$ such that $\hat \p(i)>1/(2B)$.
    \State Set $\mathcal J = \{\}$.
    \State Add the singleton interval $\{i_j\}$ to $\mathcal J$ for each $\mathcal J$.
    \For{ interval $[i_{j} + 1, i_{j+1} - 1]$ }
    \State Set $I \gets \emptyset$, $t \gets i_j + 1$
    \While{ $t < i_{j+1}$}
    \If{$\hat \p(I \cup \{t\}) < 3/(2B)$}
    \quad $I \gets I \cup \{t\}$.
    \Else
        \quad Add $I$ to $\mathcal J$.
        Start a new interval $I \gets \emptyset$.
    \EndIf
    \State Consider the next element $t \gets t + 1$.
    \EndWhile
    \EndFor
    \Return $\mathcal J$.
    \end{algorithmic}
\end{algorithm}

\begin{lemma}[Claim~1 in \cite{AcharyaDK15}] \label{lem:divide}
There exists an algorithm \textup{\algname{Approx-Divide}} that, given parameters $\delta\in(0,1]$ and integer $B > 1$, takes $O\big(B \log(B/\delta) \big)$ samples from a distribution $\p$ on $[n]$ and, with probability at least $1 - \delta$, outputs a partition of $[n]$ into $K \leq 8B$ intervals $I_1, I_2, \dots, I_K$ such that the following holds:
\begin{enumerate}
    \item For each $i \in [n]$ such that $\p(i) > 1/B$, there exists $\ell \in [K]$ such that $I_\ell = \{ i \}$. (All ``heavy'' elements are left as singletons)
    \item Every other interval is such that $\p(I) \leq 16/B$.
\end{enumerate}
\end{lemma}
\begin{proof}
Let $B>1$, and consider an arbitrary (unknown) probability distribution $\p$ over $[n]$. Let $\hat{\p}$ be the empirical distribution obtained by taking $m= \clg{18 B\ln\frac{12B}{\delta}} = O(B\log(B/\delta))$ i.i.d.\ samples from $\p$.

Denote by $i_1<\dots< i_T$ the elements $i$ such that $\p(i)>1/B$; note that $0 \leq T\leq B$ since there can be at most $B$ elements with probability mass greater than $1/B$.

Consider the following (deterministic, and unknown to the algorithm) partition of the domain into consecutive intervals:
\begin{itemize}
  \item each of $i_1<\dots< i_T$ is a singleton interval $J_{\ell_j} = \{i_j\}$;
  \item setting for convenience $i_0=-1$ and $i_{T+1}=n+1$, define the remaining intervals greedily as follows. For each $0\leq j \leq T$, to partition $\{i_{j}+1, i_{j}+2,\dots, i_{j+1}-1\}$, do
      \begin{enumerate}
        \item set $J \gets \emptyset$, $t\gets i_{j}+1$
        \item while $t < i_{j+1}$
          \begin{itemize}
            \item if $\p(J\cup\{t\}) < 3/(2B)$, set $J \gets J\cup\{t\}$
            \item else add $J$ as an interval to the partition, then start a new interval: $J\gets \emptyset$
            \item consider the next element: $t\gets t+1$
          \end{itemize}
      \end{enumerate}
\end{itemize}
That is, every remaining interval (``in-between'' two heavy elements $i_j, i_{j+1}$) is greedily divided from left to right into maximal subintervals of probability mass at most $3/(2B)$. Note that since every element $i\in \{i_{j}+1,\dots, i_{j+1}-1\}$ has mass at most $1/B$, this is indeed possible and leads to a partition of $\{i_{j}+1,\dots, i_{j+1}-1\}$ in intervals, where all but at most one (the rightmost one) has probability mass $\p(J)\in [1/(2B), 3/(2B)]$.

Overall, this process (which the algorithm cannot directly run, not knowing $\p$) guarantees the existence of a partition $J_1,\dots J_S$ of the domain in at most
\[
  S \leq T + T + 2B \leq  4B
\]
consecutive, disjoint intervals 
(the $T$ singleton intervals, the at most $T$ ``rightmost, in-between'' intervals which may have probability mass at most $1/(2B)$, and the remaining ``in-between'' intervals which all have mass at least $1/(2B)$, and of which there can thus be at most $2B$ in total).

Now, consider the following partition $\hat{J}_1,\dots \hat{J}_K$ defined by the algorithm from $\hat{\p}$:
\begin{itemize}
  \item every element $i\in[n]$ such that $\hat{\p}(i)>1/(2B)$ is added as a singleton interval; let their number be $T'$;
  \item then the algorithm proceeds as in the above process (which defined $J_1,\dots J_S$), but using $\hat{\p}$ instead of $\p$.
\end{itemize}
The same analysis (but with $T' \leq 2B$ instead of $T\leq B$, since the threshold for ``singletons'' is now $1/(2B)$) shows that this will result in a partition of the domain into $K$ consecutive disjoint intervals $\hat{J}_1,\dots \hat{J}_K$, with 
\[
  K \leq 8B\,.
\]
It remains to argue about the properties of $\hat{J}_1,\dots \hat{J}_K$, by using those of (the unknown) $J_1,\dots J_S$. First, we have that, for every $1\leq j\leq T$,
by a Chernoff bound,
\[
  \probaOf{\hat{\p}(i_j) \leq \frac{1}{2B} } \leq e^{-(1/2)^2\frac{m}{2B}} \leq \frac{\delta}{3 B}
\]
from our setting our $m$. Now, for any interval $J_j$ (where $1\leq j\leq S$) such that $1/(2B) \leq \p(J_j) \leq 3/(2B)$ (there are at most $2B$ of them), we also have, all by a Chernoff bound, 
\[
  \probaOf{\hat{\p}(i_j) \notin \big[\frac{1}{4B}, \frac{2}{B}\big] } \leq 2e^{-\frac{m}{18B}} \leq \frac{\delta}{6 B}
\]
while for any $J_j$ such that $\p(J_j) < 1/(2B)$ (there are at most $T$ of them)
\[
  \probaOf{\hat{\p}(i_j) \geq 1/B } \leq e^{-\frac{m}{6B}} \leq \frac{\delta}{3 B}\,.
\]
This means, by a union bound, that with probability at least $1-\mleft(T\cdot\frac{\delta}{3B}+2B\cdot \frac{\delta}{6 B} + T\cdot\frac{\delta}{3B}\mright) \geq 1-\delta$:
\begin{itemize}
  \item all of the $(1/B)$-heavy elements for $\p$ will also be $(1/(2B))$-heavy elements for $\hat{\p}$, and thus constitute their own singleton interval, as desired;
  \item all of the intervals $J_j$ such that $1/(2B) \leq \p(J_j) \leq 3/(2B)$ (call those ``heavy'') satisfy $1/(4B) \leq \hat{\p}(J_j) \leq 2/B$;
  \item all of the intervals $J_j$ such that $\p(J_j) < 1/(2B)$ (call those ``light'') satisfy $\hat{\p}(J_j) \leq 2/B$.
\end{itemize}
Conditioning on this happening, we can analyze the properties of $\hat{J}_1,\dots \hat{J}_K$. Specifically, fix any $\hat{J}_j$ which is not a singleton of the form $\{i_\ell\}$ for one of the $T$ $(1/B)$-heavy elements for $\p$. We want to argue that $\p(\hat{J}_j) \leq 16/B$.
\begin{itemize}
  \item If $\hat{J}_j$ intersects at most 8 $J_{j'}$'s (which must then all be consecutive, and at least 7 are ``heavy'', and the rightmost one is either heavy or light), we are done, since then $\p(\hat{J}_j) \leq 8\cdot 2/B = 16/B$.
  \item Moreover, $\hat{J}_j$ cannot intersect more than 8 $J_{j'}$'s, as otherwise those are consecutive, and thus $\hat{J}_j$ must contain at least 6 ``heavy'' $J_{j'}$'s. But in that case $\hat{\p}(\hat{J}_j) \geq 6\cdot 1/(4B) = 3/(2B)$, while by our greedy construction we ensured that $\hat{\p}(\hat{J}_j) < 3/(2B)$.
\end{itemize}
This concludes the proof.
\end{proof}

\begin{algorithm}[H]
\caption{Approx-Sub-Divide} \label{alg:approx-subdivide}
    \begin{algorithmic}[1]
    \Require Sample access from $\p$ on $[n]$; A set of disjoint intervals $\I = \{I_1, \cdots, I_q \}$; failure probability $\delta$
    \State Let $L$ be the list of elements from $[n]$ formed by concatenating the intervals $I_1, \cdots, I_q$ in order.
    \State Define the measure $\tilde \w(i) = \p(L_i)$ supported on $ [ \abs{L} ] $. 
    \State Let $\w$ be the normalized distribution of $\tilde \w$.
    \State Take $C \cdot  B \log (B / \delta)$ many samples from $\w$ for a sufficiently large $C>0$. Denote the multiset of samples by $\mathcal X$.
    \Comment{One can sample from $\w$ by  performing rejection sampling from $\p$.}
    \State $\mathcal J \gets \algname{Approx-Divide}(\mathcal X, \w, B)$.
    \State $\mathcal S_i \gets \emptyset$ for $1 \leq i \leq  q$.
    \For{Interval $J = [a,b] \in \mathcal J$}
        \If{ $L(a)$ and $L(b)$ are in the same interval $I_j$  }
            \State Add the interval $[L(a), L(b)]$ to $S_j$.
        \ElsIf{ $L(a)$ is in $I_j$ and $L(b)$ is in $I_k$ }
            \State Let $e_j$ be the right end point of $I_j$.
            and $s_k$ be the left end point of $I_k$.
            \State Add the interval $ [L(a), e_j] $ to $\mathcal S_j$.
            \State Add the interval $[s_k, L(b)]$ to $\mathcal S_k$.
            \State Set $\mathcal S_t = \{ I_t \}$ for $ j < t < k$.
        \EndIf
    \EndFor
    \Return $\mathcal S_1, \cdots, \mathcal S_q$.
    \end{algorithmic}
\end{algorithm}
\begin{proof}[Proof of Lemma~\ref{lem:sub-divide}]
Following from Lemma \ref{lem:divide}, we have, with probability at least $1 - \delta$, that the intervals in $\mathcal J$ obtained from running \algname{Approx-Divide} on the distribution $\q$ satisfies
    (i) $\mathcal J$ contains at most $8B$ intervals in total.
    (ii) Every non-singleton interval $J \in \mathcal J$ satisfies $\w(J) \leq  16/B$.

For an interval $J = [a, b] \in \mathcal J$, we can obtain a set $Q \subseteq [n]$ by mapping the element with the list $L$, i.e., add $L(x)$ to $Q$ if $x \in J$. Then, by our definition of $\w$, we will have 
$
\p(Q) = \w(J) \cdot \snorm{1}{ \tilde \w } \leq  16/B \cdot \p(\I).$
Furthermore, $Q$ will either be an interval of $[n]$ itself or a union of two disjoint intervals of $[n]$. 
In the former case, the algorithm simply adds the interval $[L(a), L(b)]$ to $\mathcal S_t$ if $[L(a), L(b)]$ is a sub-interval of $I_t$.

The latter case could only happen when $L(a) \in I_j$ and $L(b) \in I_k$ for $j < k$. The algorithm is simply adding the two intervals to $\mathcal S_j$ and $\mathcal S_k$ accordingly. As a result, the mass of each individual sub-interval still has its mass bounded by $16/B \cdot \p(\I)$ and we add at most $8B + q$ intervals to $\mathcal S_t$s.

Finally, note that we can sample from $\q$ via rejection sampling with on average $\normone{\tilde\w}^{-1} =1 / \p(\I)$ samples from $\p$. Thus, if we take $m \eqdef 2\normone{\tilde\w}^{-1} m^\ast$ samples from $\p$ (where $m^\ast \eqdef C B \log(B/\delta)$), we will get an expected $2m^\ast$ samples from $\q$; moreover, the number of samples $M$ from $\q$ follows a Binomial distribution with parameters $m$ and $\normone{\tilde\w}$. To run \algname{Approx-Divide}, we need at least $m^\ast$ samples from $\q$; the probability to obtain fewer, by a Chernoff bound, is at most
\[
	\Pr[ M < \frac{1}{2}\mathbb{E}[M] ] \leq e^{-\frac{1}{8}\mathbb{E}[M]} = e^{-\frac{m^\ast}{4}}
\]
which is at most $\delta$ as long as $m^\ast \geq 4\log(1/\delta)$ (which is satisfied by our setting, as $B\geq 1$). 
Hence, we have the algorithm takes at most $O( B \log(B / \delta)  / \p(\I))$ samples with probability at least $1 - \delta$.
\end{proof}

\subsection{Detailed Analysis of \textbf{Learn-And-Sieve}: Lemma~\ref{lem:sieve}} \label{appendix:learn-and-sieve}
We now turn to the proofs deferred from Section~\ref{ssec:sieve}, 
and establish the guarantees of Algorithm~\ref{alg:raw_sieve} as stated in Lemma~\ref{lem:sieve}.

\begin{proof}[Proof of Lemma~\ref{lem:sieve}]
The second point is easy to see since otherwise Line \ref{line:more-than-k} would reject.
We claim that the following holds with probability at least $1 - \delta$. If $\p \in \mathcal H_k$, $\Bad$ contains all the $\eps$-bad intervals and no non-breakpoint intervals. Conditioned on that, it is easy to see $\Bad$ would contain no more than $k$ intervals since there are at most $k$ break-point intervals. Besides, for any $I \in \G \eqdef \I \backslash \B$, we have $\dchi{I}{\p}{\hat \p} \leq \eps^2/K$. Overall, it then holds $\dchi{ \mathcal I  \backslash \mathcal B }{\p}{\hat{\p}} \leq \eps^2$. This concludes the proof of the second point.

It remains to show why $\Bad$ separates the bad intervals from the rest. We will separately establish that with high probability $\Bad$ will not include any non-breakpoint intervals and that it will include all bad intervals separately, and conclude by a union bound. The following statement, which follows from Lemma~\ref{lem:interval-mass-estimate} , will be useful in both aspects of the analysis. In particular, it holds
\begin{align}
&\abs{ \pmap(Q) - \p(Q)  } \leq \sqrt{ \p(Q)/b } \label{eq:tilde-Q-diff}\\
&\pmap (Q) \leq 3 \cdot \p(Q) \label{eq:tilde-Q-ratio}
\end{align}
for all sub-intervals $Q$ with probability at least $1 - \delta/2$. We will condition on that in the analysis.

\paragraph{Exclusion of non-breakpoint intervals.} By Lemma \ref{lem:learn:modified}, if $I$ is not a break-point interval, then for every sub-interval $Q$,  with probability at least $1 - \delta/3$, we have
\begin{align}
    \abs{ \hat{\p}(Q) - \p(Q)  } &\leq \sqrt{ \p(Q)/b }  \label{eq:learn-diff-guarantee} \\
    \frac{\p(Q)}{ \hat{\p}(Q)} &\leq  \max( 2, 8 n / b) 
    \label{eq:learn-ratio-guarantee}
    \, ,
\end{align}
As a result, we have
\[
(\pmap(Q)- \hat{\p}(Q))^2 \leq 4 \cdot \p(Q)/b
\leq \max\lp(2/b, 8n/b^2\rp) \cdot  \hat{\p}(Q)
\leq \frac{1}{4} \hat{\p}(Q) \cdot \eps^2/K \, ,
\]
where the first inequality follows from
Equations~\ref{eq:learn-diff-guarantee}, \ref{eq:tilde-Q-diff}, and the triangle inequality; the second inequality follows from the bounded ratio between $\hat{\p}(Q)$ and $\p(Q)$ (Equation~\ref{eq:learn-ratio-guarantee}) and the last inequality follows from $b \gg   \sqrt{nK}/\eps + K/\eps^2 $. This guarantees that the second condition of Line~\ref{rej:1} will be false.

Next, we shift our focus to the first condition of Line~\ref{rej:1}. 
Combining Equations~\ref{eq:learn-ratio-guarantee}and~\ref{eq:tilde-Q-ratio}, we have 
\[
\frac{\pmap(Q)}{ \hat \p (Q) }
= 
\frac{\pmap(Q)}{ \p(Q) } \cdot \frac{\p(Q)}{ \hat \p(Q)}
\leq  6\max \lp( 1, 4  n/b \rp)
\leq  6\max \lp( 1, \eps \sqrt{n/K} \rp).
\]
where the last inequality follows from $b \gg \sqrt{nK} / \eps$.
This guarantees that Line \ref{rej:1} will not pass either.
Hence, with probability at least $1 - \delta$, $\Bad$ will not include any non-breakpoint intervals.

\paragraph{Inclusion of $\eps$-bad intervals.} Let $I$ be an $\eps$-bad interval. Then, there exists a sub-interval $Q$ such that
$\p$ is constant within $Q$ and $\frac{\lp( \p(Q) - \hat{\p}(Q) \rp)^2}{\hat{\p}(Q)}
\geq \eps^2/K$, which implies that
\begin{align} \label{eq:diff-lower-bound}
\abs{\p(Q) - \hat{\p}(Q)} \geq \sqrt{ \hat \p(Q) \cdot \frac{\eps^2}{K} }.
\end{align}
We then proceed to analyze the cases $\p(Q) \leq \frac{2\eps}{\sqrt{Kn}}$ and $\p(Q) >  \frac{2\eps}{\sqrt{Kn}}$ separately. 
 
In the former case, we have 
\[
\frac{\p(Q)}{b} = \frac{\hat \p(Q)}{b} \cdot \frac{\p(Q)}{\hat \p(Q)}
\leq \hat \p(Q) \cdot 4\eps \sqrt{\frac{n}{K}} \cdot \frac{1}{b}
\leq \frac{1}{4} \hat \p(Q) \cdot \frac{\eps^2}{K} \, ,
\]
where the first inequality follows from that $\p(Q) \leq 2 \eps /\sqrt{nK}$ and $\hat \p(Q) \geq \frac{1}{2n}$, and the second inequality follows from $b \gg \sqrt{nK} / \eps$.
This together with Equations~\ref{eq:tilde-Q-diff}, \ref{eq:diff-lower-bound}, and the triangle inequality gives
\[
\abs{ \hat \p(Q) - \pmap(Q) } \geq \frac{1}{2} \sqrt{ \hat \p(Q) \cdot \frac{\eps^2}{K} } \, ,
\]
showing that the second condition of Line~\ref{rej:1} will evaluate to true.

In the latter case, since $\p(Q) > \frac{2\eps}{\sqrt{nK}} \gg \frac{1}{b}$, it holds that $\frac{\p(Q)}{b} \leq \frac{1}{4}\p(Q)^2$.
From Equation~\ref{eq:tilde-Q-diff} we then have
\begin{equation} \label{eq:tilde-bound}
    \pmap(Q) \geq \p(Q) - \sqrt{ \p(Q)/b }
    \geq \frac{1}{2} \p(Q).
\end{equation}
If $\pmap(Q) / \hat \p(Q) \geq 6\max(1, \eps \sqrt{n/K})$, the first condition of Line~\ref{rej:1} is true, implying $I$ will be included in $\B$.
Otherwise, if $\pmap(Q) / \hat \p(Q) < 6\max(1, \eps \sqrt{n/K})$ we get
\[
\frac{\p(Q)}{b} = \frac{ \hat \p(Q) }{b} \cdot \frac{ \p(Q) }{ \pmap(Q) } \cdot \frac{ \pmap(Q) }{ \hat \p(Q) }
< \frac{ \hat \p(Q) }{b} \cdot 2 \cdot 6\max(1, \eps \sqrt{n/K})
\leq \frac{1}{4} \hat \p(Q) \cdot \frac{\eps^2}{K} \, ,
\]
where the first inequality follows from Equation~\ref{eq:tilde-bound}, and the second inequality follows from $b \gg \sqrt{nK}/\eps + K/\eps^2$. Exactly as in the first case, we can then conclude that
\[
\abs{ \hat \p(Q) - \pmap(Q) } \geq \frac{1}{2} \sqrt{ \hat \p(Q) \cdot \frac{\eps^2}{K} } \, ,
\]
which shows that in this case the second condition of Line \ref{rej:1} will evaluate to true.
\end{proof}

\section{Lower Bound: Construction of Moment-Matching Random Variables} \label{appendix:var-construct}
Here we provide the details of the (explicit) construction of a pair of random variables 
satisfying the requirements of Lemma~\ref{lem:rv-pair}. 
We will rely on the following standard fact about polynomials:  %
\begin{fact} \label{fact:poly-moment}
Suppose $p$ is a degree-$d$ polynomial with distinct roots $r_1, \cdots , r_d$. 
Then, for every $0\leq k\leq d-2$, we have that $\sum_{i=1}^d \frac{ r_i^k }{ p'(r_i) } = 0$.
\end{fact}

\begin{proof}[Proof of Lemma~\ref{lem:rv-pair}]
We construct such a pair explicitly, relying on properties of Chebyshev polynomials. 
For every integer $d\geq 0$, recall that the corresponding 
Chebyshev polynomial (of the first kind) $T_d$ is given by 
$T_d( x ) = \cos \lp( d \cdot \arccos (x) \rp)$ for $x\in[-1,1]$ 
(and can be shown to be a polynomial of degree $d$). 
Then, let $\Delta \eqdef \frac{\sqrt{kn}}{C \cdot \log^2 n}$ 
and $d \eqdef c \cdot \log n$ (where $c,C>0$ are two absolute constants, suitably large), 
and consider the polynomial 
\begin{align}
p(x) = x \left(x - \frac{1}{n}\right)\left(x - \frac{2}{n}\right) T_d\lp( 1 -  \Delta \cdot x  \rp)\, ,
\end{align}
for $x\in[0,1/\Delta]$. 
This is a degree-$(d+3)$ polynomial, whose roots are listed in the table below along with the corresponding values of its derivative.
\begin{table}[H]
\begin{tabular}{|lll|}
\hline
Roots & \vline & $p'$  \\ \hline
$r_0 = 0$     &  \vline & 
$\frac{2}{n^2} \cdot T_d(1) = \frac{2}{n^2}$ \\
$r_1 = \frac{1}{n}$     & \vline & $- \frac{1}{n^2} \cdot T_d\bigg(1 - \frac{\Delta}{n} \bigg)$ \\
$r_2 = \frac{2}{n}$     & \vline & $ \frac{2}{n^2} \cdot T_d\bigg(1 - \frac{2\Delta}{n} \bigg)$ \\
$r_{2+m} = \frac{1}{\Delta} \lp( 1 - \cos 
\lp( \frac{2m-1}{2d} \pi \rp)\rp)$ for $1\leq m \leq  d$ & \vline &
$\abs{p'(r_{2+m})} = \Theta\mleft( \frac{m^5}{\Delta^2 d^4}  \mright)$ for $1\leq m \leq  d$. \\ \hline
\end{tabular}
\end{table}
The roots in the last row are those associated with the Chebyshev Polynomial $T_d$, and the last line relies on the facts that, for $r$ such that $T_d(1-\Delta \cdot r)=0$, we have
\[
p'(r) = -\Delta r(r-1/n)(r-2/n)T'_d(1-\Delta \cdot r)\,
\]
and that
$
T'_d(\cos \theta) = \frac{d\sin(d \cdot \theta)}{\sin\theta}
$. This latter identity implies that
$
|T'_d(1-\Delta r_{2+m})| = \frac{d}{|\sin \frac{2m-1}{2d}|} = \Theta \mleft(\frac{d^2}{m}\mright)
$, and since $r_{2+m} = \Theta(\frac{m^2}{\Delta \cdot d^2})$ we get the claimed bound:
\begin{equation}
    |p'(r_{2+m})| = \Theta\mleft( \Delta \cdot \frac{m^2}{\Delta \cdot d^2}\cdot  \max\mleft( \frac{m^2}{\Delta \cdot d^2}, \frac{1}{n} \mright)^2 \cdot \frac{d^2}{m} \mright)
    = \Theta\mleft( \frac{m^5}{\Delta^2 \cdot d^4} \mright)
\end{equation}
Some of these $d+3$ roots have positive derivatives, while others have negative derivatives: this will tell us which ones to use for our construction of $U$, and which ones for $U'$. Namely, for root $r$, we set (1)~$\Pr(U = r)=0$ and $\Pr(U' = r) \propto \frac{1}{ p'(r) }$ if $p'(r) > 0$, and (2)~$\Pr(U = r) \propto \frac{1}{ p'(r) }$ and $\Pr(U'=r)=0$ otherwise.

We now derive some bounds on the derivatives.
Firstly, for sufficiently large choice of $C$ (compared to $c$), notice that the weights on the Chebyshev polynomial's roots are overall bounded by
\begin{align} \label{eq:sum_chebyshev}
\sum_{m=1}^d \frac{1}{p'(r_{m+2})}
= 
\Theta(1) \cdot  \Delta^2 d^4 \sum_{m=1}^d  \frac{1}{m^5} \ll nk \, ,
\end{align}
since $\Delta^2 d^4 = kn \cdot (c^2/C)^2$ and the series $\sum_{m} \frac{1}{m^5}$ is convergent. Secondly, we claim that
\begin{align} \label{eq:chebyshev_o(1)}
    T_d\lp( 1 - \Delta/n \rp), T_d\lp( 1 - 2\Delta/n \rp) \geq \cos( \pi/3) = \frac{1}{2}.
\end{align}
This is because 
$
\cos( \pi / 3d )
\leq 1 - \frac{2}{d^2}$; while $1 - \Delta/n \geq 1- \frac{1}{C\log^2 n} \geq 1-\frac{2}{d^2}$, the latter inequality again for a sufficiently large choice of $C$ (with respect to $c$). The claim then follows by noticing that $T_d$ is monotonically increasing in the region $[\pi/3d, 1]$ and $T_d(\cos( \pi / 3d )) = \cos(\pi/3)$. We then proceed to verify each property claimed in the lemma.
\begin{enumerate}
    \item By our construction, $U$ is only supported on roots with negative derivatives. Hence, it can only takes values from $r_1 = \frac{1}{n}$ and some of the roots of the Chebyshev polynomials $r_{2+m}$. Moreover, by Equations \eqref{eq:sum_chebyshev}, \eqref{eq:chebyshev_o(1)}, the probabilities are bounded respectively by (before renormalization)
    $\Pr[U = \frac{1}{n}] \propto \Omega(n^2)$ and $\Pr[U \neq \frac{1}{n}] \propto \sum_{m = 1 }^d \mathds{1} \{ p'(r_{2+m}) < 0 \} |p'(r_{m+2})| \ll nk$. It then follows that
    $\Pr[U \neq \frac{1}{n}] \ll  \frac{nk}{nk + n^2}  \leq  \frac{k}{n} $.
    \item Similarly, we have
    $
    \Pr\lp[U' = 0\rp], \Pr\lp[U' = 2/n\rp]
    \propto \Omega(n^2)
    $ and by~\eqref{eq:chebyshev_o(1)} are within a factor $2$ of each other. Combined with the fact that
    \[
    \Pr\lp[U' \notin \{0, 2/n \}\rp]
    \propto
    \sum_{m = 1 }^d \mathds{1} \{ p'(r_{2+m}) > 0 \} p'(r_{2+m}) \ll nk\,,
    \]
    this immediately yields after renormalization that
    $\Pr\lp[U' = 0\rp], \Pr\lp[U' = 2/n\rp] \geq \frac{1}{3}$.
    \item The largest values $U,U'$ can take are the largest root of the Chebyshev polynomial, which is at most 
    $\max_{m \in [d]}\frac{ 1 }{\Delta } \lp( 1 - \cos\lp( \frac{2m-1}{2d} \pi \rp) \rp)
    \leq \frac{2}{\Delta} = \frac{2C\log^2 n}{\sqrt{kn}}
    $.
    \item First, as we used earlier,  recall that by Taylor expansion of $\cos$
    \[
    r_{m+2} = \frac{ 1 }{ \Delta } \lp(  1 - \cos \lp(
    \frac{2m-1}{2d} \pi \rp) \rp)
    = \Theta\mleft( \frac{m^2}{\Delta d^2} \mright)
    \]
    Then, it holds
    \begin{align*}
    \E \lp[ U \rp] &= 
    \Pr \lp[ 
    U = \frac{1}{n}
    \rp] \cdot \frac{1}{n}
    + \sum_{m=1}^d \Pr \lp[ U = r_{m+2} \rp] \cdot r_{m+2} \\
    &\leq 
    \frac{1}{n}
    + 
    \Theta(1) \cdot \sum_{m=1}^d \frac{\Delta^2 d^4}{n^2 m^5} \cdot \frac{m^2}{\Delta d^2}  \\
    &\leq \frac{1}{n}
    + \Theta\mleft(\frac{\Delta d^2}{n^2}\mright) \cdot \sum_{m=1}^d \frac{1}{m^3}
    \leq \frac{1}{n} \lp( 1 + O(\sqrt{k/n}) \rp).
    \end{align*}
    The argument for $\E \lp[ U' \rp] = \frac{1}{n} (1 + O(\sqrt{k/n}))$ is similar.
    \item The claim follows from Fact \ref{fact:poly-moment}.
\end{enumerate}
This concludes the proof of Lemma~\ref{lem:rv-pair}.
\end{proof}

 \section{Application to Model Selection}
  \label{sec:msel}
In this section, we detail how our testing algorithm, by a standard reduction, can be used for \emph{model selection}, i.e., to select a suitable parameter $k$ in order to succinctly represent the data.
\begin{theorem}\label{theo:msel}
There exists an algorithm which, given samples from an unknown distribution $\p$ on $[n]$, an error parameter $\eps\in(0,1]$, and a failure probability $\delta\in(0,1]$, outputs a parameter $1\leq K\leq n$ such that the following holds. Denote by $1\leq k\leq n$ the smallest integer such that $\p\in\mathcal{H}_k^n$. With probability at least $1-\delta$, (1)~the algorithm takes $\wt O( \lp( \sqrt{n} / \eps^2 + k / \eps^2 + \sqrt{kn} / \eps \rp) \cdot \log(1/\delta))$ samples, (2)~$1\leq K\leq 2k$, and (3)~$\p$ is at TV distance at most $\dst$ from some $H\in\mathcal{H}_{K}^n$. Moreover, the algorithm runs in time $\operatorname{poly}(m)$, where $m$ is the number of samples taken.
\end{theorem}
Before proving the theorem, we discuss the various aspects of its statement. (1)~guarantees that, with high probability, the algorithm does not take more samples than what it would if it were \emph{given} $k$ and just needed to test whether it was the right value. 
(2) ensures that the output of the algorithm is a good approximation of the true, \emph{optimal} value $k$; that is, that the model selected is essentially as succinct as it gets. Finally, (3)~guarantees that even when the output is such that $K\ll k$, the parameter $K$ is still good: that is, approximating $\p$ by a $K$-histogram still leads to a sufficiently accurate representation (even though it is even more succinct that what the \emph{true} parameter $k$ would yield). 
\begin{proof}
  The model selection algorithm is quite simple, and works by maintaining a current ``guess'' $K$ as part of a doubling search, and using the tester of to check if the current value is good. Specifically: for $0\leq j \leq \clg{\log_2 n}$, run the testing algorithm of~\autoref{thm:main} (\algname{Divide-And-Learn-And-Sieve}) with parameters $n,\eps, k \eqdef 2^j$, and $\delta_j \eqdef \frac{\delta}{2(j+1)^2}$. If the testing algorithm returns $\textsf{Accept}$, then return $K\eqdef 2^j$ and stop; otherwise, continue the loop.
  
  Let $L\leq \clg{\log_2 n}$ be the number of iterations before the algorithm stops. The probability that all $L$ invocations of the testing algorithm behaved as they should is, by a union bound, at least
  \[
      1- \sum_{j=0}^L \delta_j  = 1- \sum_{j=0}^L \frac{\delta}{2(j+1)^2} \geq 1- \sum_{j=0}^\infty \frac{\delta}{2(j+1)^2} \geq 1-\delta.
  \]
  Condition on this being the case. Since the algorithm, as soon as $j \geq \clg{\log_2 k}$, returns $\textsf{Accept}$ (since then $\p\in\mathcal{H}_{2^j}^n$), we get that $L\eqdef \clg{\log_2 k}$, and that the output satisfies $K \leq 2k$, giving (2). 
  Moreover, if $K < k$, then by the guarantee of the testing algorithm this means that $\p$ was \emph{not} $\eps$-far from $\mathcal{H}^n_K$, showing (3). Finally, the sample complexity is then
  \[
      \sum_{j=0}^L \tilde{O}\mleft(\frac{\sqrt{n}}{\eps^2}\log \frac{1}{\delta_j} + \frac{2^j}{\eps^2}\log \frac{1}{\delta_j} + \frac{\sqrt{ 2^j \cdot n} }{ \eps }\mright)
      = \tilde{O}\mleft(\frac{\sqrt{n}}{\eps^2} + \frac{k}{\eps^2} + \frac{\sqrt{n k}}{ \eps }
      \mright)\log \frac{1}{\delta}.
  \]
  since $L = O(\log k)$. This shows (1), and concludes the proof.
\end{proof}
As a final remark, we note that the guarantee provided in (2) can be improved to the optimal $1\leq K\leq k$, by modifying slightly the above procedure. Namely, after finding some $K$ such that $1\leq K\leq 2k$ as before, one can run the testing algorithm for $K/2\leq i \leq K$ (not a binary search anymore), each time with parameters $n,\eps,i$, and $\delta_i = \delta/K$. By a union bound, this incurs an extra $\delta$ probability of failure, and an additional $\tilde{O}\mleft((\sqrt{n} / \eps^2+k / \eps^2 + \sqrt{kn} / \eps )\cdot \log(1/\delta)\mright)$ samples overall, but now the output after this second step will be guaranteed to be at most $k$ (with high probability).

\end{document}